\documentclass[lettersize,journal]{IEEEtran}
\usepackage{amsmath}
\usepackage{graphicx}
\usepackage{caption}
\usepackage{amsfonts}
\usepackage{color}
\usepackage{subcaption}
\usepackage{amsthm}
\usepackage{amssymb}
\usepackage{comment}
\usepackage{algorithm}
\usepackage{algpseudocode}
\usepackage{stfloats}

\usepackage{cuted}
\usepackage{bbm}
\usepackage{cleveref}
\usepackage[english]{babel}
\usepackage{amsthm}
\usepackage{bigints}

\usepackage{pifont}
\newcommand{\cmark}{\text{\ding{51}}}
\newcommand{\xmark}{\text{\ding{55}}}

\newtheorem{theorem}{Theorem}

\algrenewcommand\algorithmicrequire{\textbf{Input:}}
\algrenewcommand\algorithmicensure{\textbf{Output:}}

\IEEEoverridecommandlockouts
\begin{document}
\title{On the Effect of Pulse Shaping Filters in Zak-OTFS Waveform for Radar Sensing}
\author{Abhishek Bairwa and Ananthanarayanan Chockalingam\thanks{The authors are with the Department of Electrical Communication Engineering, Indian Institute of Science, Bangalore 560012, Karnataka State, India. Email: \{abhishekbair,achockal\}@iisc.ac.in.} 
}
\maketitle

\begin{abstract} 
In radar sensing, the self-ambiguity function of the probing waveform plays a crucial role in the resolvability and detection of multiple targets. A desired distribution of the ambiguity volume (volume under squared self-ambiguity function) for good sensing performance is characterized by 1) narrow main lobe, 2) low peak sidelobe ratio, and 3) low integrated sidelobe ratio. With Zak-OTFS waveform, this volume distribution is achieved by choosing appropriate delay-Doppler (DD) pulse shaping filter. In the recent Zak-OTFS based radar literature, Gaussian pulse shaping filter has been considered, and it has been shown to offer better range/velocity estimation performance compared to the traditional chirp waveform in scenes with multiple targets. While the self-ambiguity function with Gaussian filter has very low side lobes, its main lobe is wide which compromises resolvability and performance. Motivated by this, we seek filters with better ambiguity characteristics. Specifically, in this paper, we explore two other known filters, namely, sinc and Gaussian-sinc (GS) filters, and demonstrate that these filters offer better performance compared to Gaussian filter under different scenarios and receiver processing. Towards demonstrating this, we derive closed-form expressions for the self-ambiguity functions of Zak-OTFS waveform with sinc and GS filters, which have not been reported before. The ambiguity functions of sinc and GS filtered waveforms have narrow main lobes, resulting in better resolvability in scenes with densely populated targets for the basic peak-detection based receiver. The ambiguity function of Gaussian filtered waveform has very low sidelobes, resulting in better performance in sparsely populated scenes. When a receiver with inter-target interference mitigation is used, the sinc and GS filters perform better in both dense and sparsely populated scenes compared to Gaussian filter.
\end{abstract}

\section{Introduction}
\label{sec:intro}
Integrated sensing and communication (ISAC) is projected to be a prominent feature in the next generation of wireless communication systems \cite{ITU_future_trend}. Advanced waveforms that can deliver robust communication as well as good radar sensing performance will be crucial for ISAC \cite{ITU_future_trend}-\cite{ISAC_2}. One such emerging waveform is the orthogonal time frequency space (OTFS) waveform which is a pulse in the delay-Doppler (DD) domain and a pulsone in the time domain \cite{otfs1}-\cite{otfs2}. Earlier works on OTFS have focused on the multicarrier version of OTFS (MC-OTFS) for communication and radar sensing \cite{otfs_ita}-\cite{otfs_isac4}. Recently, Zak transform based OTFS (Zak-OTFS) waveform has been shown to give more robust communication performance in wireless channels characterized by large Doppler spreads \cite{zak_otfs1}-\cite{zak_otfs7}. In addition, in the recent literature, Zak-OTFS waveform has been shown to be a better radar waveform for resolvability/detection of multiple targets compared to the traditional chirp waveform \cite{zak_otfs_radar},\cite{zak_polarimetry}.

In radar sensing, self-ambiguity function of the probing waveform plays a crucial role in the resolvability and detection of multiple targets \cite{richards_book}. Formally, the self-ambiguity function is defined as the cross-correlation of the waveform with its delayed and Doppler-shifted version. The volume under the ambiguity surface can be interpreted as an inherent “blur” that limits the ability to distinguish multiple targets. According to Moyal’s identity \cite{moran}, the total volume under the squared self-ambiguity function is finite and equals the squared energy of the waveform, implying that this blur is fundamentally unavoidable. Consequently, radar waveform design focuses on distributing this volume in a manner that minimizes performance degradation.

The self-ambiguity of the traditionally used chirp waveform is supported on a line in the DD domain. To detect a single target, an up-chirp and a down-chirp signals are transmitted successively. The intersection point of lines corresponding to self-ambiguity of the up-chirp and the down-chirp gives the target's DD location, which gives the range and velocity of the target. For a radar scene consisting of multiple targets, this approach results in intersection points not corresponding to any real target, giving rise to the so called ``ghost'' targets. To distinguish ghost targets from real targets, multiple pairs of up-chirp/down-chirp are transmitted, which compromises the delay (range) and Doppler (velocity) resolution \cite{zak_otfs_radar},\cite{chirp}. 

In contrast, the self-ambiguity of a Zak-OTFS waveform is supported on a lattice in the DD domain \cite{zak_otfs3},\cite{zak_amb_swaroop}, providing better localization leading to unambiguous target detection. A filtered Zak-OTFS waveform has the volume under the ambiguity function concentrated around these lattice points. Specifically, the volume/blur around the origin plays a significant role in the detection and resolvability of multiple targets. In Zak-OTFS, the distribution of this volume is affected by the DD domain pulse shaping filter used in the waveform generation. Different pulse shaping filters will result in different distributions of the volume, and hence different sensing performances.

Recently, in \cite{zak_otfs_radar}, a Gaussian filtered Zak-OTFS waveform has been considered for radar sensing in a multi-target scene. It has been shown to achieve significantly better range/velocity estimation performance compared to the traditional chirp waveform and the MC-OTFS waveform (Fig. 10 in \cite{zak_otfs_radar}), due to the favorable characteristics of its self-ambiguity function compared to that of the chirp and MC-OTFS waveforms. Although the self-ambiguity function of a Gaussian filtered Zak-OTFS waveform around origin is more localized and has very low side lobes, its main lobe is wide. Consequently, the targets which are close to each other in the DD domain are difficult to resolve. This motivated us to seek filters with better ambiguity function characteristics that can achieve better target resolvability and detection performance, which forms the basic premise in this paper. 
 
Towards this end, in this paper, we consider two other known filters, namely, sinc and Gaussian-sinc (GS) filters \cite{zak_otfs2},\cite{gs},\cite{closed_form}, and demonstrate that these two filters offer better performance compared to the Gaussian filter under different scenarios and receiver processing. The key contributions and insights in this work can be summarized as follows.
\begin{itemize}
\item {\em Closed-form expressions for self-ambiguity functions:} 
The Zak-OTFS framework enables obtaining precise expressions for the self-ambiguity function of the waveform generated using different pulse shaping filters \cite{zak_otfs3}. Exploiting this attribute, a closed-form expression for the ambiguity function of Zak-OTFS waveform with sinc filter has been reported in the literature \cite{zak_otfs3} (Ch. 10, Th. 10.5). In this paper, we derive closed-form expressions for the ambiguity function of Gaussian and GS filtered waveforms, which are new and useful contributions. The availability of the expressions is summarized in Table \ref{tab1}. The derived expressions allow significant reduction in simulation run times compared to numerical evaluation of integrals, e.g., simulation run times get reduced by factors of about 10 and 80 for GS and Gaussian filters, respectively, compared to numerical evaluation.     

\item {\em Characteristics of the self-ambiguity functions:}
As the structure of the self-ambiguity function of a waveform plays a vital role in determining its sensing performance, we characterize the self-ambiguity functions of Zak-OTFS waveform with sinc, Gaussian and GS filters in terms of 1) main lobe width, 2) peak sidelobe ratio (PSLR)\footnote{PSLR is defined as the ratio of the peak of the highest sidelobe to the peak of the main lobe in the ambiguity function.} and 3) integrated sidelobe ratio (ISLR)\footnote{ISLR is defined as the ratio of area under the sidelobes to the area under the main lobe in the ambiguity function.}. The ambiguity function of the Gaussian filtered waveform exhibits a wider main lobe than those of the sinc and GS filters, and almost no discernible sidelobes. In contrast, the sinc and GS filters have a narrow main lobe and significant sidelobes (see Fig. \ref{fig:zero_dop_cut_self_amb}). Although both sinc and GS filters exhibit similar PSLR, they differ in terms of ISLR, with the sinc filter having a significantly higher ISLR compared to that of the GS filter. These characteristics are summarized in quantitative terms in Table \ref{tab2}. 

\item {\em Comparative sensing performance of different filters:}
We present a comparison of the sensing performance with sinc, Gaussian, and GS filters in terms of detection capability (receiver operating characteristic - ROC) and range/velocity estimation accuracy in two different radar scenes: 1) scene with densely populated targets, and 2) scene with sparsely populated targets. When a basic peak-detection based receiver is used, the sinc and GS filters provide better detection and estimation performance compared to those of the Gaussian filter in densely populated scene, due to the narrow main lobe of their ambiguity functions. In sparsely populated scenes, good detection performance is dictated by low ISLR in the ambiguity function. Thus, in sparse scenes, the Gaussian filter provides the best detection capability among the three filters, followed by the GS filter. Whereas, the high PSLRs of the sinc and GS filters result in weaker targets getting affected by stronger targets in sparse scenes, leading to poor estimation accuracy. To improve performance, we consider a receiver with inter-target interference (ITI) mitigation which involves estimation of delay, Doppler, and fade parameters associated with a target, and subsequently reconstructing and subtracting its contribution from the whole cross-ambiguity function, to aid better estimation of parameters associated with the other targets. When the ITI mitigating receiver is used, because of  effective interference estimation and mitigation enabled by their narrow main lobes, the sinc and GS filters perform better in both dense and sparsely populated scenes compared to the Gaussian filter.
\end{itemize}

\begin{table}[t]
\centering
\vspace{2mm}
\begin{tabular}{|l||c|c|c|}
\hline
DD filter type &  \multicolumn{3}{|c|}{Closed-form expression for self-ambiguity function} \\ 
\cline{2-4} 
& \multicolumn{1}{c|}{Ref. \cite{zak_otfs_radar}} & \multicolumn{1}{c|}{Ref. \cite{zak_otfs3} (Ch. 10)} & \multicolumn{1}{c|}{This paper} 
\\ \hline \hline
Sinc            & {\color{red} $\xmark$} & {\color{blue} $\cmark$} & {\color{red} $\xmark$} 
\\ \hline
Gaussian  & {\color{red} $\xmark$}     & {\color{red} $\xmark$} & {\color{blue} $\cmark$} 
\\ \hline
Gaussian-sinc  & {\color{red} $\xmark$} & {\color{red} $\xmark$} & {\color{blue} $\cmark$} 
\\ \hline
\end{tabular}
\caption{Availability of closed-form expressions for the self ambiguity functions for different filters.}
\label{tab1}
\end{table}

\begin{table}[t]
\centering
\vspace{2mm}
\begin{tabular}{|l||c|c|c|}
\hline
DD filter type &  \multicolumn{3}{|c|}{Self-ambiguity function characteristics} \\ 
\cline{2-4} 
& \multicolumn{1}{c|}{Main lobe width} & \multicolumn{1}{c|}{PSLR} & \multicolumn{1}{c|}{ISLR} 
\\ 
& \multicolumn{1}{c|}{(@ -25 dB level)} & \multicolumn{1}{c|}{} & \multicolumn{1}{c|}{} 
\\ \hline \hline
Sinc            & 1.99 &  -6.65 dB & 2.31 dB  
\\ \hline
Gaussian  & 5.39     & - & - 
\\ \hline
Gaussian-sinc  & 2.12 & -6.92 dB & -3.12 dB \\ \hline
\end{tabular}
\caption{Characteristics of the self ambiguity functions of different filters.} 
\label{tab2}
\end{table}

The rest of the paper is organized as follows. Section \ref{sec:system_model} introduces the Zak-OTFS system model for radar sensing. In Section \ref{sec:selfamb}, we derive closed-form expressions for the self-ambiguity functions of Zak-OTFS waveforms filtered through different pulse shaping filters. In Section \ref{sec:effect}, we highlight the effect of pulse shaping filters on estimation accuracy in different radar scenes and present an ITI mitigation scheme to improve performance. Section \ref{sec:results} presents the performance results and discussions. Conclusions and possible future work are presented in Section \ref{sec:concl}.

\section{Zak-OTFS system model for radar sensing}
\label{sec:system_model}

\begin{figure*}
\centering    
\includegraphics[width=0.8\linewidth]{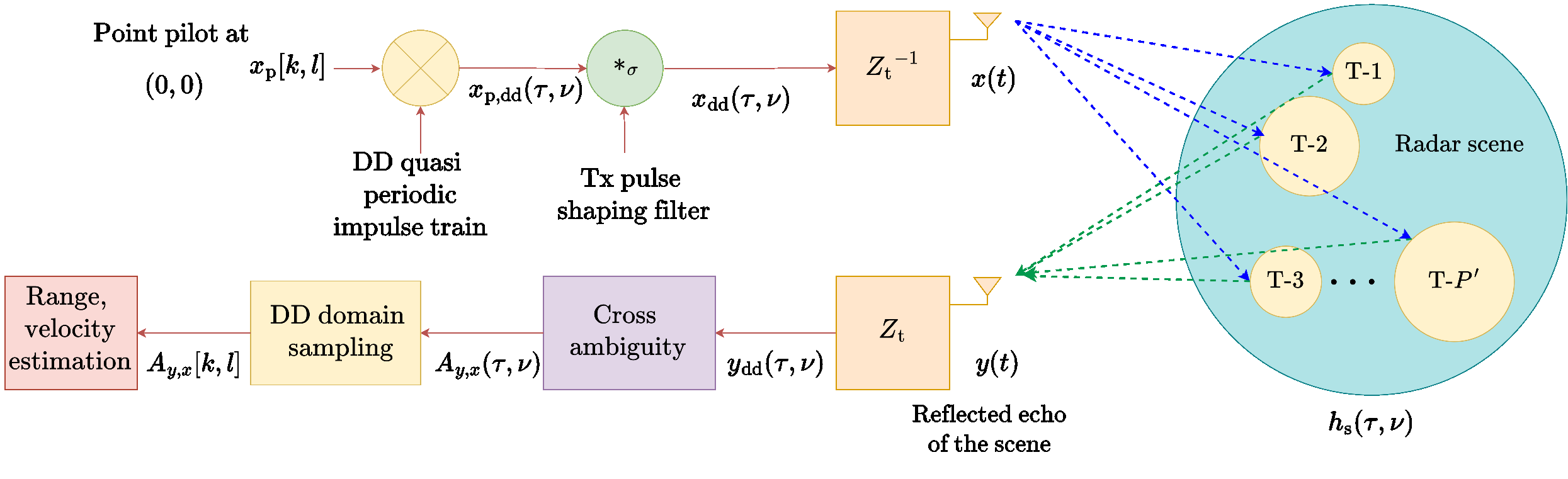}
\vspace{1mm}
\caption{Block diagram of Zak-OTFS transceiver for radar sensing.}
\label{Block_diagram}      
\vspace{-4mm}
\end{figure*}

In a Zak-OTFS based radar system, a pulse in the DD domain is used as the probing waveform for sensing. A pulse in the DD domain is a quasi-periodic function, parameterized by delay period $\tau_{\text{p}}$ and Doppler period $\nu_{\text{p}}$, with $\tau_{\text{p}}\nu_{\text{p}}=1$. The fundamental region $\mathcal{D}_{0}$ in the DD domain is defined as $\mathcal{D}_{0}  = \left\{(\tau,\nu)\ | \ 0\leq\tau<\tau_{\text{p}}, \ 0\leq\nu<\nu_{\text{p}}\right\}$, where $\tau$ and $\nu$ are the delay and Doppler variables, respectively. $\mathcal{D}_{0}$ is divided into $M$ bins along the delay axis and $N$ bins along the Doppler axis, where $M = B\tau_{\text{p}}$ and $N = T\nu_{\text{p}}$ for a probing waveform limited to bandwidth $B$ and time duration $T$. In order to generate the probing waveform, a unit energy impulse function in the DD domain, i.e., $x_{\text{p}}[k,l]$ = $\delta[k,l]$, $0 \leq k< M$, $0 \leq l <N$, where $\delta[.]$ is the Kronecker delta function, is converted into a quasi-periodic discrete DD domain signal $x_{\text{p,dd}}[k,l]$ using the following encoding:
\begin{eqnarray}
    x_{\text{p,dd}}[k+nM,l+mN] = x_{\text{p}}[k,l]e^{j2\pi n \frac{l}{N}}, \ n,m\in\mathbb{Z}. 
    \label{eq1}
\end{eqnarray}
The discrete DD domain signal in (\ref{eq1}) is converted into a continuous DD domain signal, given by
\begin{eqnarray}
   x_{\text{p,dd}}(\tau,\nu)=\sum_{k,l\in \mathbb{Z}} x_{\text{p,dd}}[k,l] \delta\Big(\tau-\frac{k\tau_{\text{p}}}{M}\Big)\delta\Big(\nu-\frac{l\nu_{\text{p} }}{N}\Big), 
   \label{eq2}
\end{eqnarray}
where $\delta(.)$ is the Dirac-delta function. The continuous DD domain signal in (\ref{eq2}) occupies infinite time duration and bandwidth. Thus, in order to limit the time duration and bandwidth of the probing waveform, a transmit DD domain pulse shaping filter $w_{\text{tx}}(\tau,\nu)$ is applied to $x_{\text{p,dd}}(\tau, \nu)$. The pulse shaped continuous DD domain signal is given by  
\begin{eqnarray}
\label{eq3}
    x_{\text{dd}}(\tau,\nu) = w_{\text{tx}}(\tau,\nu) *_{\sigma} x_{\text{p,dd}}(\tau, \nu),
\end{eqnarray}
where $*_{\sigma}$ denotes the twisted convolution\footnote{$a(\tau,\nu)*_{\sigma}b(\tau,\nu) = \int\int a(\tau',\nu')b(\tau-\tau',\nu-\nu')e^{j2\pi\nu'(\tau-\tau')}d\tau'd\nu'.$}. The probing waveform $x(t)$ is the time domain realization of $x_{\text{dd}}(\tau, \nu)$, obtained as $x(t) = {Z_{t}}^{-1}(x_{\text{dd}}(\tau, \nu))$ where $Z_{t}^{-1}$ denotes the inverse Zak transform\footnote{$Z_{t}^{-1}(a(\tau,\nu)) =\sqrt{\tau_{\text{p}}}\int_{0}^{\nu_{\text{p}}}a(t,\nu)d\nu.$}. The probing signal $x(t)$ is transmitted from the transmit antenna to sense the radar scene whose DD spreading function is given by 
\begin{eqnarray}
\label{eq4}
    h_{\text{s}}(\tau,\nu) = \sum_{i =1}^{P'}h_{i}\delta(\tau-\tau_{i})\delta(\nu-\nu_{i}),
\end{eqnarray}
where $P'$ is the number of targets in the scene, $\tau_{i}$, $\nu_{i}$, and $h_{i}$ are the delay, Doppler, and complex fade associated with the $i$th target. 
The probing waveform interacts with the radar scene, and the echo signal $y(t)$ is received by the receive antenna co-located with the transmit antenna. The received echo signal is given by 
\begin{eqnarray}
    y(t) = \int\int h_{\text{s}}(\tau,\nu)x(t-\tau)e^{j2\pi\nu(t-\tau)}d\tau d\nu + n(t),
\end{eqnarray}
where $n(t)$ is the additive white Gaussian noise (AWGN) with variance $N_{0}$, i.e., $\mathbb{E}[n(t)n^{*}(t-t')] = N_{0} \delta(t')$. At receiver, the echo signal $y(t)$ is converted into a DD domain signal, given by $y_{\text{dd}}(\tau,\nu) = Z_{t}(y(t))$ where $Z_{t}$ is the Zak transform\footnote{$Z_{t}(a(t)) = \sqrt{\tau_{\text{p}}}\Sigma_{k\in\mathbb{Z}}a(\tau+k\tau_{\text{p}})e^{-j2\pi\nu k\tau_{\text{}p}}$.}. The cross-ambiguity between DD domain signals $y_{\text{dd}}(\tau,\nu)$ and $x_{\text{dd}}(\tau,\nu)$ is 
given by \cite{zak_otfs_radar}
\begin{eqnarray}
\label{eq6}
\hspace{-1mm}
    A_{y,x}(\tau,\nu) & \hspace{-2mm} = & \hspace{-2mm} \int y(t)x^{*}(t-\tau)e^{-j2\pi\nu(t-\tau)} dt \nonumber \\ & \hspace{-34mm} \overset{(a)}{=} & \hspace{-19mm}  \int_{0}^{\tau_{\text{p}}} \hspace{-1mm} \int_{0}^{\nu_{\text{p}}} y_{\text{dd}}(\tau',\nu')x_{\text{dd}}^{*}(\tau'-\tau,\nu'-\nu)e^{-j2\pi\nu(\tau'-\tau)}d\tau'd\nu' \nonumber \\&\hspace{-34mm}\overset{(b)}{=}& \hspace{-18mm} h_{\mathrm{s}}(\tau,\nu)*_{\sigma}A_{x,x}(\tau,\nu) + A_{n,x}(\tau,\nu) \nonumber \\
    & \hspace{-34mm} \overset{(c)}{=} & \hspace{-18mm} \sum_{i =1}^{P'}h_{i}A_{x,x}(\tau - \tau_{i}, \nu-\nu_{i})e^{j2\pi\nu_{i}(\tau-\tau_{i})} + A_{n,x}(\tau,\nu),
\end{eqnarray}
where step (a) follows from Eq. (10.9) in \cite{zak_otfs3}, 
step (b) follows from Theorem 10.1 in \cite{zak_otfs3},
and step (c) follows from substituting (\ref{eq4}) in step (b).
Here, $A_{x,x}(\tau,\nu)$ denotes the self-ambiguity function of the probing waveform $x(t)$, given by
\begin{eqnarray}
\label{eq:self_amb_t}
A_{x,x}(\tau,\nu) & \hspace{-2mm} = & \hspace{-2mm} \int x(t)x^{*}(t-\tau)e^{-j2\pi\nu(t-\tau)}dt  \\ \nonumber &\hspace{-34mm}=&\hspace{-19mm} \int_{0}^{\tau_{\text{p}}} \hspace{-1mm} \int_{0}^{\nu_{\text{p}}} \hspace{-1mm} x_{\text{dd}}(\tau',\nu')x_{\text{dd}}^{*}(\tau'-\tau,\nu'-\nu)e^{-j2\pi\nu(\tau'-\tau)}d\tau'd\nu',     
\end{eqnarray}
which is further simplified by substituting (\ref{eq3}) and following Theorem 10.4 in \cite{zak_otfs3} to obtain the self-ambiguity expression in (\ref{eq:self_amb_DD}) given at the top of the next page.
\begin{figure*}[t]
\begin{eqnarray}
    \label{eq:self_amb_DD}
  A_{x,x}(\tau,\nu) 
  & = &  \sum_{n \in \mathbb{Z}} \sum_{m\in \mathbb{Z}} w_{\text{tx}}(\tau,\nu) *_{\sigma}\left(e^{j2\pi m\nu_{\text{p}}\tau}w_{\text{tx}}^{*}( n\tau_{\text{p}} - \tau, m\nu_{\text{p}}-\nu)e^{j2\pi(\tau-n\tau_{\text{p}})(\nu-m\nu_{\text{p}})}\right).
\end{eqnarray}
\hrule
\end{figure*}
$A_{n,x}(\tau,\nu)$ represents the cross-ambiguity between $n(t)$ and $x(t)$, given by 
\begin{eqnarray}
\label{eq:Anx}
A_{n,x}(\tau,\nu) = \int n(t)x^{*}(t-\tau)e^{-j2\pi\nu(t-\tau)}dt.   
\end{eqnarray}
The covariance of noise in the cross-ambiguity domain is given by (see Appendix \ref{app:derive_noise_cov} for the derivation) 
\begin{eqnarray}
\label{eq:noise_cov_amb}
\mathbb{E}[A_{n,x}(\tau_{1},\nu_{1})A^{*}_{n,x}(\tau_{2},\nu_{2})] \nonumber\\& \hspace{-28mm}=  N_{0}e^{-j2\pi\nu_{2}(\tau_{2}-\tau_{1})}A_{x,x}(\tau_{1}-\tau_{2},\nu_{1}-\nu_{2}).
\end{eqnarray}

The cross-ambiguity $A_{y,x}(\tau,\nu)$ is sampled on a DD domain lattice finer than the delay and Doppler resolution. Specifically, a discrete cross-ambiguity function is obtained as $A_{y,x}[k,l] = A_{y,x}(\tau = \frac{k}{PB},\nu = \frac{l}{QT})$, where $k,l \in \mathbb{Z}$, and $P$ and $Q$ are positive integer constants representing the oversampling factors along the delay and Doppler domains, respectively. The delay and Doppler parameters of the targets can be estimated through peak detection in the discrete cross-ambiguity function, i.e., 
\begin{equation} 
(\hat\tau_{i},\hat\nu_{i}) = \left(\frac{\hat k_{i}}{PB},\frac{\hat l_{i}}{QT}\right), 
\end{equation}
where $(\hat k_{i}, \hat l_{i})$, $i = 1,2,\cdots,P'$, denotes the estimated peak locations in $A_{y,x}[k,l]$ corresponding to the $P'$ targets. The corresponding estimated ranges and velocities of the targets are given by 
\begin{equation}
(\hat r_{i}, \hat v_{i}) = \left(\frac{c \hat \tau_{i}}{2}, \frac{c \hat \nu_{i}}{2f_{c}}\right),
\end{equation} 
where $c$ is the speed of light and $f_{c}$ is the carrier frequency.

{\em Choice of $\tau_{\mathrm{p}}$ and $\nu_{\mathrm{p}}$}:  The self-ambiguity function $A_{x,x}(\tau,\nu)$ is supported on the period lattice $\Lambda_{\text{p}}$, given by ${\Lambda}_{\text{p}} \triangleq \left\{(n\tau_{\text{p}}, m\nu_{\text{p}})| n,m \in \mathbb{Z}\right\}$ as shown in Fig. \ref{fig:heatmaps_self_amb}. In the cross-ambiguity function $A_{y,x}(\tau,\nu)$, the lattice points other than $(0,0)$ act as DD aliases which can result in ghost targets while detecting the peaks. To prevent the DD aliases interference in the detection of target peaks corresponding to lattice point $(0,0)$, the delay and Doppler periods are chosen so that the aliases are separated well, i.e., $\tau_{\text{p}} > (\tau_{\text{max}}-\tau_{\text{min}})$ and $\nu_{\text{p}} > (\nu_{\text{max}}-\nu_{\text{min}})$, which is referred to as the {\em crystallization condition} for radar sensing using Zak-OTFS \cite{zak_otfs_radar}. Here, $\tau_{\text{min}}$ and $\tau_{\text{max}}$ are the minimum and maximum delays, and $\nu_{\text{min}}$ and $\nu_{\text{max}}$ are the minimum and maximum Dopplers, of the targets.

{\em Peak detection DD window}: The crystallization condition ensures separation of peaks corresponding to ghost and real targets. Consequently, this enables the choice of an appropriate DD window to detect the peaks corresponding to all the real targets. To achieve this, the delay spread of the peak detection window should lie between $(\tau_{\text{max}} - \tau_{\text{min}})$ and $\tau_{\text{p}}$. Similarly, the Doppler spread of the window should lie between $(\nu_{\text{max}} - \nu_{\text{min}})$ and $\nu_{\text{p}}$. Specifically, we define the peak detection window by the co-ordinates $\{ (\tau'_{\text{min}}, \nu'_{\text{min}}), (\tau'_{\text{max}}, \nu'_{\text{min}}),(\tau'_{\text{min}}, \nu'_{\text{max}}),(\tau'_{\text{max}}, \nu'_{\text{max}})\}$. Here, $\tau'_{\text{max}}$ and $\nu'_{\text{max}}$ are chosen to be greater than $\tau_{\text{max}}$ and $\nu_{\text{max}}$, by some delay and Doppler margins, respectively. These margins are chosen to account for the DD spreading of the self-ambiguity (which depends on the pulse shaping filter used). Similarly, $\tau'_{\text{min}}$ and $\nu'_{\text{min}}$ are chosen to be lesser than $\tau_{\text{min}}$ and $\nu_{\text{min}}$, by the same delay and Doppler margins, respectively.

{\em Implementation}: In practical implementation of the receiver (Algorithm 1 in \cite{zak_otfs_radar}), the Zak-transforms of the time domain echo signal and the probing waveform are obtained for discrete DD grid points in the fundamental region $\mathcal{D}_{0}$ using fast Fourier transform (FFT). The discrete cross-ambiguity between the echo and probing signals is obtained in a discrete peak detection DD window using Riemann sum as an approximation for continuous integral, and subsequently the peak location are identified to give range/velocity estimation.

\subsection{DD pulse shaping filters} 
We consider pulse shaping filters of the form $w_{\text{tx}}(\tau,\nu) = w_{1}(\tau)w_{2}(\nu)$, for which the probing waveform is given by \cite{zak_otfs3},\cite{zak_otfs_radar}
\begin{eqnarray}
  x(t) = Z_{t}^{-1}(x_{\text{dd}}(\tau,\nu)) = w_{1}(t)  \star  {\big(W_{2}(t) \ p(t)\big)},   
\end{eqnarray}
where $\star$ denotes linear convolution,
$p(t)$ is the point pulsone in the time domain, given by 
\begin{equation}
p(t) = Z_{t}^{-1}(x_{\text{p,dd}}(\tau,\nu)) = \sqrt{\tau_{\text{p}}} \sum_{n\in\mathbb{Z}}\delta(t-n\tau_{\text{p}}), 
\end{equation}
and $W_2(t) \triangleq \int w_{2}(\nu)e^{j2\pi\nu t}d\nu$. To limit the duration of the pulsone to $T$, the duration of $W_2(t)$ is limited to $T$. Similarly, the bandwidth of the pulsone is limited to $B$ by limiting the bandwidth of $w_{1}(t)$ to $B$. We consider sinc, Gaussian and Gaussian-sinc (GS) pulse shaping filters as described below.

{\em Sinc filter:} The  
sinc pulse shaping filter is given by
\begin{eqnarray}
\label{eq:wtx_sinc}
   w_{\text{tx}}(\tau,\nu) = \sqrt{B}\text{sinc}(B\tau)\sqrt{T}\text{sinc}(T\nu). 
\end{eqnarray}
For the sinc filter, the time duration of probing waveform $T'$ and bandwidth of the waveform $B'$ are $T$ and $B$, i.e., no time and bandwidth expansion is incurred. 

{\em Gaussian filter:} The  
Gaussian 
filter is given by  
\begin{eqnarray}
\label{eq:wtx_gauss}
  w_\text{tx}(\tau,\nu) \hspace{-0.5mm} = \hspace{-0.5mm} {\left(\frac{2\alpha_{\tau}B^2}{\pi}\right)^{\frac{1}{4}}\hspace{-0.8mm}e^{-\alpha_{\tau}B^{2}\tau^{2}}} \ {\left(\frac{2\alpha_{\nu}T^2}{\pi}\right)^{\frac{1}{4}}\hspace{-0.8mm}e^{-\alpha_{\nu}T^{2}\nu^{2}}},  
\end{eqnarray}
where $\alpha_{\tau}$ and $\alpha_{\nu}$ are adjustable parameters. As the Gaussian pulse has infinite support, the time duration $T'$ and bandwidth $B'$ are defined by ${99}\%$ energy containment. In order to have no time/bandwidth expansion, i.e., $T' = T, \ B' = B$, the parameters are set to $\alpha_{\tau} = \alpha_{\nu} = 1.584$. 

{\em Gaussian-sinc filter:} The 
GS filter is given by
\begin{equation}
\label{eq:wtx_GS}
   w_\text{tx}(\tau,\nu)\hspace{-0mm}=\hspace{-0mm}\Omega_\tau\Omega_\nu\sqrt{BT} \mathrm{sinc}(B\tau)\mathrm{sinc}(T\nu) e^{-\alpha_{\tau}B^{2}\tau^{2}} \hspace{-1mm} e^{-\alpha_{\nu}T^{2}\nu^{2}}\hspace{-1mm}. 
\end{equation}
The roll-off parameters are set to $\alpha_{\tau} = \alpha_{\nu} = 0.044$ for no time/bandwidth expansion and the energy normalization parameters are $\Omega_{\tau} = \Omega_{\nu} = 1.0278$ \cite{gs}. 

\section{Self-ambiguity function for different filters} 
\label{sec:selfamb}
In this section, we study the self-ambiguity functions of the Zak-OTFS waveform obtained using different pulse shaping filters. Towards this, we 1) present heatmaps of the self-ambiguity functions of sinc, Gaussian, and GS filters, highlighting the supporting lattice structure and the spread of the ambiguity volume, and 2) derive closed-form expressions for the self-ambiguity functions for these filters. 

\begin{figure*}[t]
\centering
\includegraphics[width = 16cm, height = 4cm]{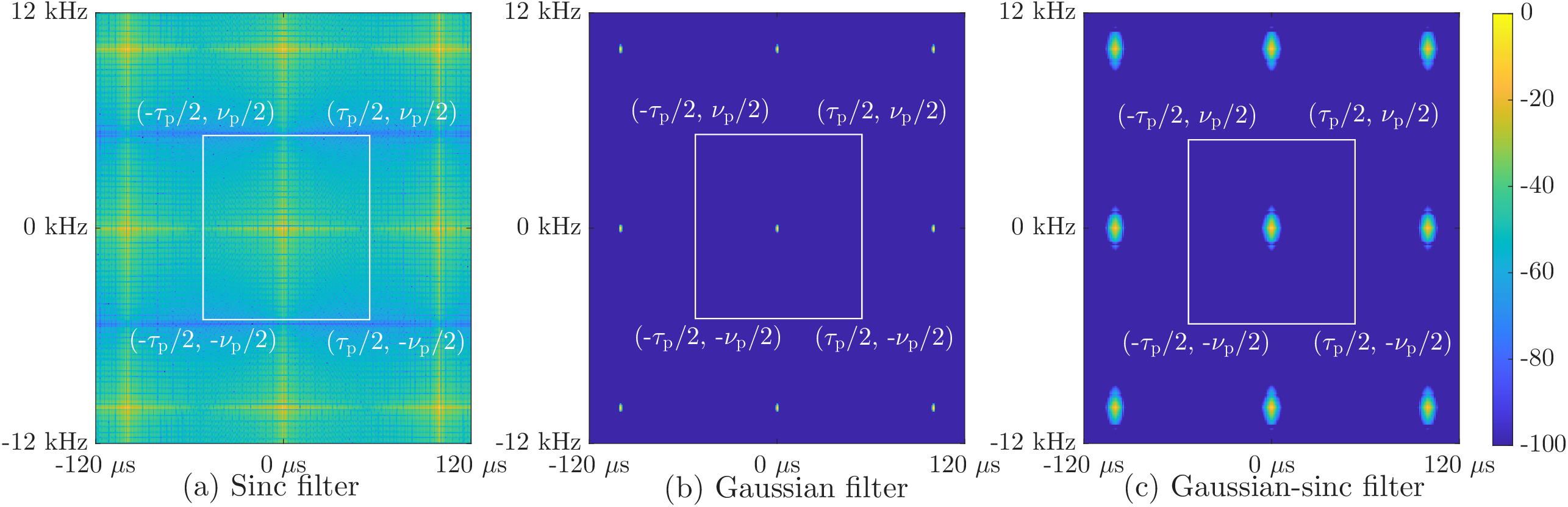}
\caption{Heatmaps of the self-ambiguity functions of Zak-OTFS waveform filtered through sinc, Gaussian, and GS filters. }
\label{fig:heatmaps_self_amb}
\end{figure*}

\begin{figure}
\centering
\includegraphics[width = 8cm, height = 5.3cm]{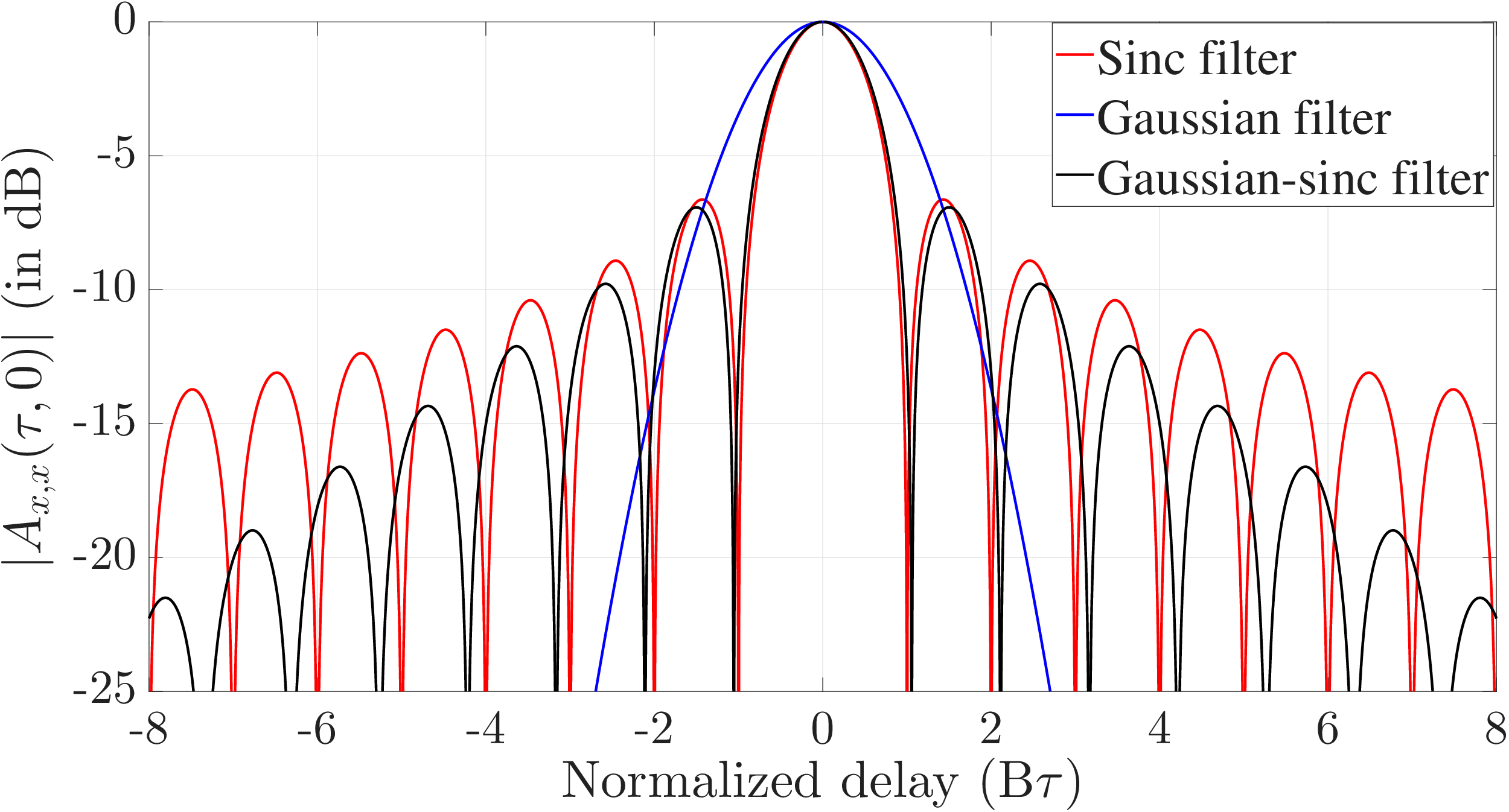}
\caption{Zero-Doppler cut of magnitude of self-ambiguity functions of Zak-OTFS waveform for different filters. }
\label{fig:zero_dop_cut_self_amb}
\end{figure}

Figure \ref{fig:heatmaps_self_amb} shows the heatmaps of the self-ambiguity functions corresponding to different filters. We consider a delay period of $\tau_{\text{p}} = 100$ $\mu s$, and so the Doppler period $\nu_{\text{p}} = 1/\tau_{\text{p}} = 10$ kHz. From Fig. \ref{fig:heatmaps_self_amb}, we observe that the self-ambiguity functions are supported on the period lattice with lattice points at $( n\tau_{\text{p}}, m\nu_{\text{p}})$ for $n,m \in \mathbb{Z}$, and the volume spread around these lattice points is influenced by the pulse shaping filter. For the sinc filter, the volume around the lattice points is widely spread along the delay and Doppler dimensions. In contrast, the spread is very much localized around the lattice points in the case of the Gaussian filter. For the GS filter, the spread is more contained than that of the sinc filter but less than that of the Gaussian filter.

In Fig. \ref{fig:zero_dop_cut_self_amb}, we plot the magnitude of the zero-Doppler cut of the self-ambiguity functions, i.e., $|A_{x,x}(\tau, \nu = 0)|$ in dB scale, corresponding to different filters as a function of normalized delay ($B \tau$) around $\tau = 0$. The ambiguity for the Gaussian filtered waveform has a wide main lobe, i.e., at $-25$ dB point, the main lobe width is 5.39 delay bins, each bin of size $1/B$, and almost no discernible sidelobes. For the sinc filter, the ambiguity function has a narrow main lobe but has high side lobe levels. Specifically, the main lobe width is 1.99 delay bins  at $-25$ dB point, the PSLR is -6.65 dB, and the ISLR is 2.31 dB in the delay range $[-\tau_{\text{p}}/2,\tau_{\text{p}}/2]$ as indicated by the white box in the Fig. \ref{fig:heatmaps_self_amb}. For the GS filter, the ambiguity function has a narrow main lobe and a PSLR similar to those of the sinc filter, but the ISLR is much less than that of the sinc filter. Specifically, the main lobe width is 2.12 delay bins at -25 dB point, the PSLR is -6.92 dB, and the ISLR is -3.12 dB in the delay range of $[-\tau_{\text{p}}/2,\tau_{\text{p}}/2]$. The above characteristics of the main lobe and sidelobes for different filters are summarized in Table \ref{tab2}. Similar characteristics can be observed for the zero-delay cut of the self-ambiguity functions.  

\subsection{Closed-form expressions}  
In this subsection, we derive closed-form expressions for the self-ambiguity functions of Zak-OTFS waveform with Gaussian and GS filters (Theorems 1 and 2), which have not been reported in the literature. For the sinc filter, a closed-form expression has been reported in \cite{zak_otfs3} (Ch. 10, Th. 10.5), which we reproduce here (Theorem 3) for immediate reference purposes.   

\begin{theorem}
\label{thm:self_amb_gauss}
    The closed-form expression for the self-ambiguity function of Gaussian filtered Zak-OTFS waveform is given by
    \begin{eqnarray}
        A_{x,x}(\tau,\nu) = \sum_{n \in \mathbb{Z}} \sum_{m \in \mathbb{Z}} \left(e^{-\alpha_\tau B^{2} t_{0}^{2}} \ e^{j2\pi\nu t_{0}} \ e^{2\alpha_{\tau} B^{2}K_{1}^{2}}\right) \nonumber\\ \left( e^{-\alpha_{\nu}T^{2}f_{0}^{2}} \ e^{\alpha_{\nu}T^{2}\frac{K_{2}^{2}}{2}}\right),
    \end{eqnarray}
    where $t_{0} = \tau-n\tau_{\text{p}}, \ f_{0} = \nu-m\nu_{\text{p}}, \ K_{1} = \left( \frac{j2\pi\nu-2\alpha_{\tau}B^{2}t_{0}}{4\alpha_{\tau}B^{2}} \right)^{2}$ and $K_{2} = \left( f_{0} + \frac{j2\pi n \tau_{\text{p}}}{2\alpha_{\nu}T^{2}}\right) $.
\end{theorem}

\begin{IEEEproof}
See Appendix \ref{app:proof_self_amb_gauss}.
\end{IEEEproof}

\begin{theorem}
\label{thm:self_amb_GS}
    The closed-form expression for the self-ambiguity function of GS filtered Zak-OTFS waveform is given by
    \begin{eqnarray}
        A_{x,x}(\tau,\nu) \hspace{-2mm}&= \hspace{-2mm}& \sum_{n \in \mathbb{Z}} \sum_{m \in \mathbb{Z}} \left( I_{1,1}(\tau,\nu) \mathbb{I}_{\{t_{0}\neq0\}}  +  I_{1,2}(\tau,\nu)\mathbb{I}_{\{t_{0}=0\}}\right) \nonumber \\&&\hspace{3mm}  \left(I_{2,1}(\tau,\nu) \mathbb{I}_{\{f_{0\neq0}\}}  +  I_{2,2}(\tau,\nu)\mathbb{I}_{\{f_{0}= 0\}}\right),
    \end{eqnarray}
   where $t_{0} = \tau-n\tau_{\text{p}}, \ f_{0} = \nu-m\nu_{\text{p}} \ $. The functions $I_{1,1}(\tau,\nu)$, $I_{1,2}(\tau,\nu)$, $I_{2,1}(\tau,\nu)$ and $I_{2,2}(\tau,\nu)$ are defined as follows:
   \begin{eqnarray}
   \hspace{-10mm}
       I_{1,1}(\tau,\nu) & \hspace{-2mm} = & \hspace{-2mm} P_{1}e^{-\frac{\pi^{2}\nu^{2}}{2\alpha_{\tau}B^{2}}}e^{j\pi\nu t_{0}} \nonumber\\
       &&\hspace{-16mm}\begin{split}
       \left[ e^{j\pi B t_{0}}e^{\frac{\pi^2 K_{1}^2}{2\alpha_{\tau}B^{2}}}g_{1}\left(\frac{\pi^2}{2\alpha_\tau B^{2}}, B+K_{1}, K_{1}\right)\right.\\
       - \left. e^{-j\pi B t_{0}} e^{\frac{\pi^2 K_2^{2}}{2\alpha_\tau B^2}} g_{1}\left(\frac{\pi^2}{2\alpha_\tau B^2}, B+K_{2}, K_{2}\right) \right.\\
       \left. +  e^{j\pi B t_{0}}e^{\frac{\pi^2K_{2}^2}{2\alpha_\tau B^2}}g_1\left(\frac{\pi^2}{2\alpha_\tau B^2}, K_2, -B+K_2\right)\right.\\ \left. - e^{-j\pi B t_{0}}e^{\frac{\pi^2K_1^2}{2\alpha_{\tau} B^2}}g_1\left(\frac{\pi^2}{2\alpha_{\tau}B^2}, K_1, -B+K_1\right)\right], 
       \end{split}
   \end{eqnarray}
   where $t_{0} = \tau-n\tau_{\text{p}}$, $P_{1} = \frac{\Omega_{\tau}^2}{j2\pi t_{0} B} \sqrt{\frac{\pi}{2\alpha_{\tau}B^2}}e^{-\frac{\alpha_\tau B^2 t_{0}^2}{2}}$, $K_{1} = \frac{\alpha_{\tau}B^2}{\pi^2}\left(\frac{\pi^2 \nu}{\alpha_{\tau} B^2} + j\pi t_{0}\right)$, $K_{2} = \frac{\alpha_{\tau}B^2}{\pi^2}\left(\frac{\pi^2 \nu}{\alpha_{\tau} B^2} - j\pi t_{0}\right)$, and the function $g_1(.)$ is defined as $g_1(a,t,s) = \int_{s}^{t} e^{-ax^2}dx = \frac{1}{\sqrt{a}}\left[\frac{\sqrt{\pi}}{2}\left(erf(\sqrt{a } \ t )-erf(\sqrt{a}\ s)\right)\right]$, where $a,t,s$ are complex numbers and erf(.) is the error function.

 \begin{eqnarray}
 \hspace{-8mm}
   I_{2,1}(\tau,\nu) & \hspace{-2mm} = & \hspace{-2mm} \frac{\Omega_{\nu}^2}{j2\pi f_{0}T}\sqrt{\frac{\pi}{2\alpha_{\nu}T^2}}e^{-\frac{\alpha_{\nu}T^2 f_{0}^2}{2}} \nonumber \\
   && \hspace{-16mm}
   \begin{split}
       \left[ e^{j2\pi f_{0}\left(\frac{T}{2} + n\tau_{\text{p}}\right)} f\left(T+n\tau_{\text{p}}, n\tau_{\text{p}}, \pi f_{0}, \frac{\pi^2}{2\alpha_{\nu}T^2}\right)\right. \\ 
       \left. - e^{-j2\pi f_{0}\left(\frac{T}{2}\right)}f\left(T+ n\tau_{\text{p}}, n\tau_{\text{p}}, -\pi f_{0}, \frac{\pi^2}{2\alpha_{\nu}T^2}\right)\right.\\
       \left. - e^{j2\pi f_{0}\left(-\frac{T}{2} + n \tau_{\text{p}}\right)}f\left(n\tau_{\text{p}}, -T + n\tau_{\text{p}}, \pi f_{0}, \frac{\pi^2}{2\alpha_{\nu}T^2} \right) \right.\\
       \left. + e^{j2\pi f_{0} \left(\frac{T}{2}\right)}f\left( n\tau_{\text{p}}, -T + n\tau_{\text{p}}, -\pi f_{0}, \frac{\pi^2}{2\alpha_{\nu}T^2}\right)\right], 
   \end{split}
   \end{eqnarray}
    where the function $f(.)$ is defined as $f(t,s,z,a) = \int_{s}^{t}e^{-ax^2}e^{-jzx}dx = \frac{\sqrt{\pi}}{2}\frac{e^{-z^2/4a}}{\sqrt{a}}\left[erf\left(\sqrt{a} \ t + j\frac{z}{2\sqrt{a}}\right)-erf\left(\sqrt{a} \ s + j\frac{z}{2\sqrt{a}}\right)\right]$~for complex numbers $t,s,z,a$.
   
   \begin{eqnarray}
   \hspace{-6mm}
       I_{1,2}(\tau,\nu) & \hspace{-2mm} = & \hspace{-2mm} \frac{\Omega_{\tau}^2}{B}\sqrt{\frac{\pi}{2\alpha_\tau B^2}} \nonumber\\ & & \hspace{-18mm}
       \begin{split}       
        \left[ g_2\left(\frac{\pi^2}{2\alpha_\tau B^2}, \nu, -B+\nu\right)
        -g_2\left(\frac{\pi^2}{2\alpha_\tau B^2}, B+\nu,\nu\right) \right.\\
        \left.+(B+\nu)g_{1}\left(\frac{\pi^2}{2\alpha_\tau B^2}, B+\nu,\nu\right)\right.\\
        \left.+ (B-\nu)g_1\left(\frac{\pi^2}{2\alpha_\tau B^2}, \nu,-B+\nu\right)\right],     
        \end{split}
   \end{eqnarray}
   where the functions $g_{1}(.)$ and $g_{2}(.)$ are~defined as $g_1(a,t,s) = \int_{s}^{t} e^{-ax^2}dx = \frac{1}{\sqrt{a}}\left[\frac{\sqrt{\pi}}{2}\left(erf(\sqrt{a } \ t )-erf(\sqrt{a}\ s)\right)\right]$~and $g_{2}(a,t,s) = \int_{s}^{t}xe^{-\frac{\pi^2 x^2}{2 \alpha_{\tau} B^2}}dx = \frac{\alpha_\tau B^2}{\pi^2}\left(e^{-\frac{\pi^2}{2\alpha_\tau B^2}s^2}- e^{-\frac{\pi^2}{2\alpha_\tau B^2}t^2}\right)$ for complex numbers $a,t,s$.
   \begin{eqnarray}
   \hspace{-4mm}
       I_{2,2}(\tau,\nu) &=& \frac{\Omega_{\nu}^2}{T}\sqrt{\frac{\pi}{2\alpha_\nu T^2}} \nonumber \\ && \hspace{-20mm}
       \begin{split}
         \left[ g_{2}\left(\frac{\pi^2}{2\alpha_{\nu}T^2}, n\tau_{\text{p}}, -T + n\tau_{\text{p}}\right)\right. \\
           \left. - g_{2}\left(\frac{\pi^2}{2\alpha_{\nu}T^2}, T + n\tau_{\text{p}}, n\tau_{\text{p}}\right)\right.\\
            \left. +(T+ n\tau_{\text{p}})g_{1}\left(\frac{\pi^2}{2\alpha_{\nu}T^2}, T + n\tau_{\text{p}}, n\tau_{\text{p}}\right)\right.\\
           \left. +\left(T-n\tau_{\text{p}}\right)g_{1}\left(\frac{\pi^2}{2\alpha_{\nu}T^2}, n\tau_{\text{p}}, -T+n\tau_{\text{p}}\right)\right],         
       \end{split}
   \end{eqnarray}
    where the functions $g_{1}(.)$ and $g_{2}(.)$ are defined as the same as in the case of $I_{1,2}(\tau,\nu)$.
\end{theorem}

\begin{IEEEproof}
See Appendix \ref{app:proof_self_amb_GS}.
\end{IEEEproof}

The closed-form expression for the self-ambiguity function of sinc filtered Zak-OTFS waveform is given in \cite{zak_otfs3} (Ch.10 Th. 10.5), which is presented here for immediate reference purposes. 
\begin{theorem}
The closed-form expression for the self-ambiguity function of sinc filtered Zak-OTFS waveform is given as follows: For even $N$, 
\begin{eqnarray}
        A_{x,x}(\tau,\nu) = \frac{(1-\frac{|\nu|}{B})}{N}\sum_{n=-\frac{N}{2}}^{\frac{N}{2}-1}\sum_{m=-\frac{N}{2}}^{\frac{N}{2}-1}[e^{j\pi\nu(\tau - (n-m)\tau_{\text{p}})}\nonumber\\\mathrm{sinc}((B-|\nu|)(\tau+(n-m)\tau_{\text{p}}))]{\large\mathbb{I}_{|\nu|<B}},
\end{eqnarray}
where $\mathbb{I}$ is the indicator function. For odd $N$, the summation for both $n$ and $m$ variables starts from $-(N-1)/2$ and ends at $(N-1)/2$.
\end{theorem}

\begin{IEEEproof}
See Appendix 10.K in \cite{zak_otfs3}.
\end{IEEEproof}

{\em Remark on simulation run time:} The closed-form expressions derived in the above help in significantly accelerating the simulation run times. For example, generation of $100$ realizations of the self-ambiguity functions in a DD window of $[-\tau_{\text{p}}/2, \tau_{\text{p}}/2] \times[-\nu_{\text{p}}/2, \nu_{\text{p}}/2]$ (as indicated in Fig. \ref{fig:heatmaps_self_amb}) for sinc, Gaussian and GS filters using numerical integration took $12.62$s, $1502.8$s, and $8269.1$s, respectively. Whereas, using the closed-form expressions, the run times were observed to be $13.2$s, $19.2$s, and $842.8$s for sinc, Gaussian and GS filters, respectively\footnote{The simulations are run on a PC with a 13th Gen Intel Core i7-13700 processor (16 cores, 24 threads) and 31.1 GiB of RAM, using MATLAB R2025b.}. That is, simulation speed ups by factors of about 80 and 10 for Gaussian and GS filters, respectively, are achieved by the closed-form expressions.

\section{Effect of pulse shaping filters and ITI Mitigation}
\label{sec:effect}
In this section, we illustrate the difference in the sensing performance of different filtered Zak-OTFS waveforms due to the different characteristics of their self-ambiguity functions. We consider range/velocity estimation performance in an illustrative experiment with two targets in the high SNR regime and evaluate the performance as a function of separation between the two targets. We also demonstrate that an inter-target interference (ITI) mitigation scheme applied on the cross-ambiguity output enhances the achieved performance.

\begin{figure}
\centering
\includegraphics[width = 7.5cm, height = 6cm]{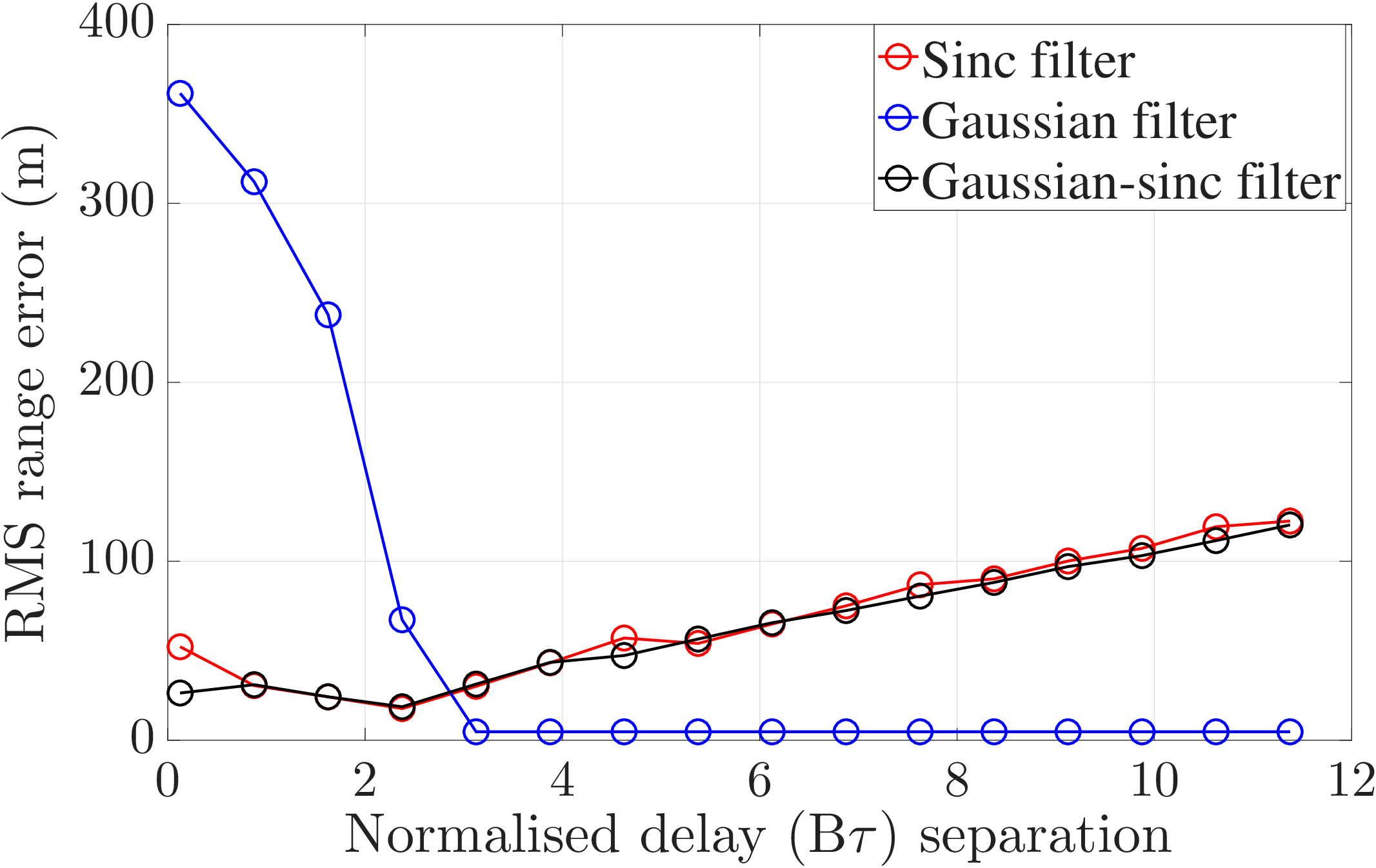}
\caption{Range estimation error vs normalized delay separation between the two targets in high SNR regime.}
\label{fig:exp_range_err}
\end{figure}

\subsection{An illustrative two-target experiment}
To illustrate the effect of pulse shaping on sensing performance, we consider a scene consisting of two targets separated in the delay domain, i.e., $\tau_{1}\neq\tau_{2}$ but sharing the same Doppler shifts, i.e., $\nu_{1}=\nu_2=0$ Hz, in the high SNR regime, where $(\tau_1,\nu_1)$ and $(\tau_2,\nu_2)$ are the DDs of targets 1 and 2, respectively. We take bandwidth $B=4$ MHz, time duration $T = 20$ ms, delay period $\tau_{\text{p}} = 100\ \mu s$, carrier frequency $f_{c} = 1$ GHz, and oversampling factors $P=Q=4$. We fix the delay shift of the first target at $\tau_1 = 0.25$ $\mu\text{s}$ and vary the delay shift of the second target $\tau_2$. The peak detection DD window is considered to be $[0, 5]$ $\mu$s $\times$ $[-1000, 1000]$ Hz. In Fig. \ref{fig:exp_range_err}, we plot the range estimation error for the two targets obtained for sinc, Gaussian, and GS filters as a function of the separation between the targets. We observe that when the targets are close, the Gaussian filtered waveform gives high range estimation errors in contrast to the sinc and GS filtered waveforms, which give much better accuracy in estimation. On the other hand, when the targets are well separated, the sinc and GS filtered waveforms give higher estimation errors compared to that of the Gaussian filtered waveform. This behavior in estimation performance is attributed to the respective self-ambiguity functions which is illustrated via the heatmaps in Fig. \ref{fig:heatmaps_toyeg}. The Gaussian filtered waveform has a broader peak/wide main lobe in its self-ambiguity function, and hence the targets cannot be distinguished as different peaks when they are close to each other, as shown in Fig. \ref{fig:heatmaps_toyeg}(b). For sinc and GS filters, the ambiguity functions have narrower peaks (with significant energy leaked into the side lobes) and therefore the close targets can be distinguished as different peaks. In the case of well separated targets, due to high side lobes in the ambiguity functions of sinc and GS filtered waveforms, the weaker target (target 2) is affected  by the side lobe of the stronger target (target 1). Consequently, it cannot be detected as a separate peak as shown in Figs. \ref{fig:heatmaps_toyeg}(a) and (c), respectively, and the accuracy in estimation suffers. For the Gaussian filtered waveform, the weaker target (target 2) is identified as a peak because of very low side lobes associated to the stronger target (target 1).

\begin{figure*}[t]
\centering
\includegraphics[width = 16cm, height = 4cm]{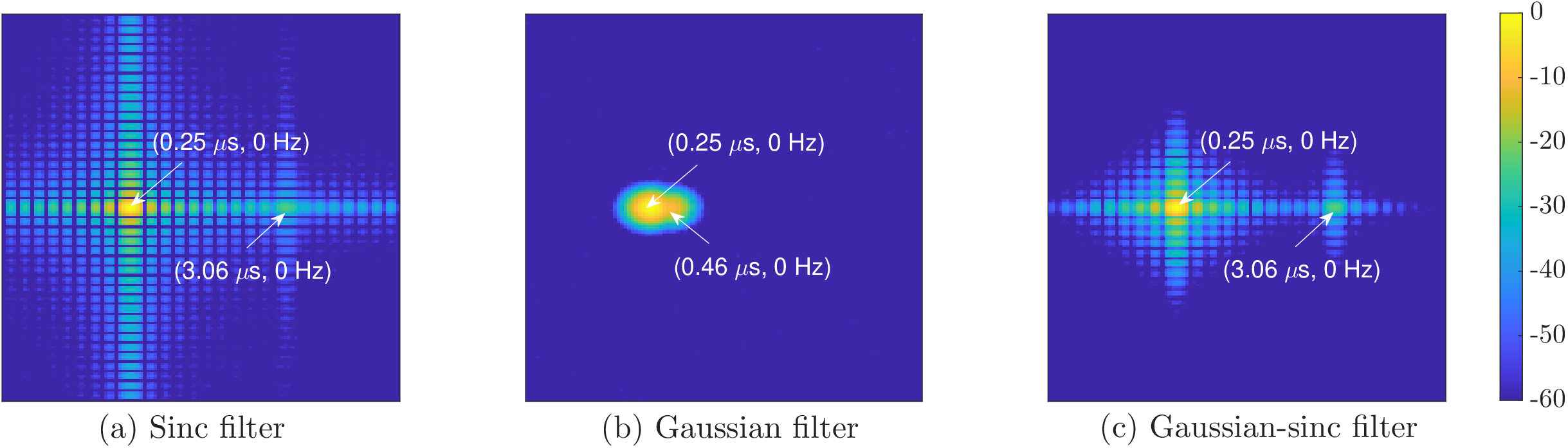}
\caption{Heatmaps illustrating the effect of sinc, Gaussian, and GS filters on the sensing performance in a two-target scene.} 
\label{fig:heatmaps_toyeg}
\end{figure*}

\subsection{Inter-target interference mitigation}
The sensing performance achieved by the basic processing of the cross-ambiguity (peak detection) can be further enhanced by additional DD signal processing. Towards this, we recognize that while the narrow main lobe is a positive attribute of the sinc and GS filters, which makes them good for dense target distribution (compared to Gaussian filter), the high side lobes are detrimental, and if the sidelobe effects are mitigated properly, their performance can be made good for a more general target distribution. With this motivation, we consider an ITI mitigation scheme, which works as follows. 

\begin{figure*}[t]
\begin{eqnarray}
\label{eq:estim_h1}
   \langle A_{y,x}(\tau,\nu), A_{x,x}(\tau -\hat\tau_1,\nu - \hat\nu_1)e^{j2\pi\hat\nu_1(\tau-\hat\tau_1)} \rangle &=& \int_{\tau'}\int_{\nu'}A_{y,x}(\tau',\nu')A^{*}_{x,x}(\tau'-\tau_1,\nu'-\nu_1)e^{-j2\pi\hat\nu_1(\tau-\hat\tau_1)}d\tau'd\nu' \nonumber \\
   & \hspace{-10cm}\overset{(a)}{=} & \hspace{-5cm} h_1 \underbrace{\int_{\tau'}\int_{\nu'} A_{x,x}(\tau'-\tau_1,\nu'-\nu_1)A^{*}_{x,x}(\tau'-\hat\tau_1,\nu'-\hat\nu_1)e^{j2\pi\nu_1(\tau-\tau_1)}e^{-j2\pi\hat\nu_1(\tau-\hat\tau_1)}d\tau'd\nu'}_{\text{1. Inner product of ambiguity functions for same target.}} \nonumber \\
   & \hspace{-10cm} & \hspace{-5cm} + \ h_2\underbrace{\int_{\tau'}\int_{\nu'}  A_{x,x}(\tau'-\tau_2,\nu'-\nu_2)A^{*}_{x,x}(\tau'-\hat\tau_1,\nu'-\hat\nu_1)e^{j2\pi\nu_2(\tau-\tau_2)}e^{-j2\pi\hat\nu_1(\tau-\hat\tau_1)}d\tau'd\nu'}_{\text{2. Inner product of ambiguity functions for different targets.}} \nonumber \\
   & \hspace{-10cm} & \hspace{-5cm} + \
   \underbrace{\int_{\tau'}\int_{\nu'} A_{n,x}(\tau',\nu')A^{*}_{x,x}(\tau'-\hat\tau_1,\nu'-\hat\nu_1)e^{-j2\pi\hat\nu_1(\tau-\hat\tau_1)}d\tau'd\nu'}_{\text{3. Noise term}}.
\end{eqnarray}
\hrule
\end{figure*}

Let $x(t)$ denote the unit energy Zak-OTFS pulsone  transmitted as the probing waveform. Consider a radar scene consisting of two targets with $h_{\text{s}}(\tau,\nu) = h_{1}\delta(\tau-\tau_1)\delta(\nu-\nu_1) + h_{2}\delta(\tau-\tau_2)\delta(\nu-\nu_2)$. At the receiver, the cross-ambiguity function $A_{y,x}(\tau,\nu)$ is given by $A_{y,x}(\tau,\nu) = \sum_{i=1}^2h_{i}A_{x,x}(\tau-\tau_{i},\nu-\nu_i)e^{j2\pi\nu_{i}(\tau-\tau_{i})} + A_{n,x}(\tau,\nu)$. Let the DD location of the maximum value of $|A_{y,x}(\tau,\nu)|^2$, representing the delay and Doppler estimate of the closest/strongest target at $(\tau_1,\nu_1)$ be $(\hat\tau_{1},\hat\nu_1)$. To estimate $h_{1}$, we compute the inner product of the cross-ambiguity $A_{y,x}(\tau,\nu)$ and the estimated ambiguity for the first target in the cross-ambiguity expression, i.e., $A_{x,x}(\tau -\hat\tau_1,\nu - \hat\nu_1)e^{j2\pi\hat\nu_1(\tau-\hat\tau_1)}$, as in (\ref{eq:estim_h1}) given at the top of the next page. In (\ref{eq:estim_h1}), where step (a) follows from substituting the $A_{y,x}(\tau,\nu)$ expression. We observe that the first term in step (a) in (\ref{eq:estim_h1}), i.e., the inner product of ambiguity functions for the same target, is 1 if $\hat \tau_1 = \tau_1 $ and $\hat \nu_1 = \nu_1 $ due to the fact that $\int_{\tau'}\int_{\nu'}|A_{x,x}(\tau',\nu')|^2d\tau' d\nu' = E^2 = 1$. The second term, i.e., the inner product of ambiguity functions for different targets, is approximately 0 if the targets are well separated and the support of the self-ambiguity functions centered at the target locations do not overlap significantly. Assuming that these conditions are satisfied approximately, we can estimate $h_1$ as 
\begin{eqnarray}
\hspace{-3mm}
\hat h_{1} & \hspace{-2mm} = & \hspace{-2mm} \langle A_{y,x}(\tau,\nu), A_{x,x}(\tau -\hat\tau_1,\nu - \hat\nu_1)e^{j2\pi\hat\nu_1(\tau-\hat\tau_1)} \rangle.   
\end{eqnarray}
We reconstruct the contribution of the first target in the cross-ambiguity function, i.e., $\hat A_{y,x,1} = \hat h_{1} A_{x,x}(\tau-\hat \tau_1, \nu-\hat \nu_1)e^{j2\pi\hat\nu_1(\tau-\hat\tau_1)}$ and subtract it from $A_{y,x}(\tau,\nu)$  to obtain an ITI canceled version, i.e., 
\begin{equation}
A_{y,x,\text{tc}}(\tau,\nu) = A_{y,x}(\tau,\nu)- \hat h_1A_{x,x}(\tau-\hat \tau_1, \nu-\hat \nu_1)e^{j2\pi\hat\nu_1(\tau-\hat\tau_1)}. 
\end{equation}
From the obtained $A_{y,x,\text{tc}}(\tau,\nu)$, we find the location of the peak $(\hat \tau_2, \hat \nu_2)$ which corresponds to the estimate of the delay and Doppler shift of the second target. This procedure can be extended for more than two targets. 

{\em Processing of the sampled DD domain signal:} The continuous DD cross-ambiguity function $A_{y,x}(\tau,\nu)$ is sampled on a DD domain lattice to obtain the discrete cross-ambiguity function as $A_{y,x}[k,l] = A_{y,x}(\tau = \frac{k}{PB},\nu =~\frac{l}{QT})$,~$k,l \in~\mathbb{Z}$, and $P$ and $Q$ are the oversampling factors along delay and Doppler, respectively, for further processing to estimate ranges/velocities of the targets. To limit the computational complexity of processing and prevent the identification of ghost targets due to DD aliasing, we consider the $A_{y,x}[k,l]$ in a limited discrete peak detection DD window given by 
\begin{eqnarray}
\mathcal{S} & \hspace{-2mm} \triangleq & \hspace{-2mm} \big\{(k,l) \ | \  k,l \in \mathbb{Z}, P\lfloor B\tau'_{\text{min}} \rfloor \leq k \leq P \lceil B\tau'_{\text{max}} \rceil, \nonumber \\
& \hspace{-2mm} & \hspace{-2mm} Q\lfloor T\nu'_{\text{min}} \rfloor \leq l \leq Q \lceil T\nu'_{\text{max}} \rceil \big\},
\end{eqnarray}
where $\tau'_{\text{max}}$ and $\nu'_{\text{max}}$ are chosen to be greater than the maximum delay $\tau_{\text{max}}$ and the maximum Doppler $\nu_{\text{max}}$ of the targets, by some delay and Doppler margins, respectively. The margins are chosen to account for the DD spreading of the self-ambiguity. Whereas, $\tau'_{\text{min}}$ and $\nu'_{\text{min}}$ are chosen to be lesser than $\tau_{\text{min}}$ and $\nu_{\text{min}}$, by the same delay and Doppler margins, respectively. 

The interference mitigation scheme is applied on the sampled DD function $A_{y,x}[k,l]$ in the DD window $\mathcal{S}$ via discrete DD signal processing as follows. The location of the peak $(\hat k_1, \hat l_1)$ in the over-sampled cross-ambiguity function gives the estimate of delay and Doppler shift parameters of the first (closest/strongest) target, i.e., $(\hat \tau_1 = \frac{\hat k_1}{PB}, \hat \nu_1 = \frac{\hat l_1}{QT})$. We calculate the volume under the squared discrete self-ambiguity function centered at the estimated first target location, in the DD window $\mathcal{S}$ as 
\begin{equation}
\hspace{-2mm}
V = \frac{1}{k_\text{s}l_\text{s}} \sum_{k'= P\lfloor B \tau'_{\text{min}}\rfloor}^{P\lceil B \tau'_{\text{max}}\rceil}\sum_{l'= Q\lfloor T \nu'_{\text{min}}\rfloor}^{Q\lceil T \nu'_{\text{max}}\rceil}\left|A_{x,x}[k'-\hat k_1, l'-\hat l_1]\right|^2\hspace{-2mm},
\end{equation}
where $k_{\text{s}}$ is the discrete delay spread over $\mathcal{S}$, given by $k_\text{s} = P\lceil B \tau'_{\text{max}}\rceil- P\lfloor B \tau'_{\text{min}}\rfloor +1$ and $l_{\text{s}}$ is the discrete Doppler spread over $\mathcal{S}$, given by $l_{\text{s}} = Q\lceil T \nu'_{\text{max}}\rceil -  Q\lfloor T \nu'_{\text{min}}\rfloor +1$. The estimate of $h_{1}$ is obtained as the inner product of the discrete cross-ambiguity function and the discrete self-ambiguity function centered at the estimated location of the first target over the window $\mathcal{S}$, normalized by the volume $V$, i.e., 
\begin{eqnarray}
\hspace{-5mm}\hat h_{1} &\hspace{-2mm}= &\hspace{-2mm} \frac{1}{k_{\text{s}}l_{\text{s}}V}\sum_{k'= P\lfloor B \tau'_{\text{min}}\rfloor}^{P\lceil B \tau'_{\text{max}}\rceil}\sum_{l'= Q\lfloor T \nu'_{\text{min}}\rfloor}^{Q\lceil T \nu'_{\text{max}}\rceil}A_{y,x}[{k'},{l'}] \nonumber\\
  &\hspace{-2mm}&\hspace{-2mm}A_{x,x}^{*}[{k'-\hat k_1},{l'-\hat l_1}]e^{-j2\pi\frac{\hat l_{1}(k'-\hat k_1)}{PQMN}}. 
\end{eqnarray}
We reconstruct the contribution of the first target and subtract it from the cross-ambiguity function to obtain 
the ITI canceled version, i.e., 
\begin{eqnarray}
\hspace{-6mm}
A_{y,x,\text{tc}}[k,l] & \hspace{-2mm} = & \hspace{-2mm} A_{y,x}[k,l] \nonumber \\
& \hspace{-2mm} & \hspace{-2mm} - \hat h_{1}A_{x,x}[k-\hat k_1,l-\hat l_1]e^{j2\pi}\frac{\hat l_1(k-\hat k_1)}{PQMN}.
\end{eqnarray}
The location of the peak in $A_{y,x,\text{tc}}[k,l]$ gives the estimate of delay and Doppler shift parameters of the second target. This procedure can be extended for more than 2 targets.

{\em Heatmaps without and with ITI mitigation:} 
In Fig. \ref{fig:heatmap_iti}, we demonstrate the effectiveness of the ITI mitigation scheme described above in the two-target scene via heatmaps of the cross-ambiguity functions obtained before and after ITI cancellation. For this illustration, we consider the scenarios where Gaussian and sinc filtered waveforms perform poorly (without ITI cancellation) in terms of estimation accuracy. Accordingly, for the Gaussian filtered waveform, we consider the case of {\em closely spaced targets}, where the performance is poor without ITI mitigation. Figures \ref{fig:heatmap_iti}(a) and (b) show the cross-ambiguity function heatmaps for the Gaussian filter before and after ITI mitigation, respectively. Due to the wide main lobe of the ambiguity function, the targets are not identified as separate peaks without ITI mitigation (Fig. \ref{fig:heatmap_iti}(a)). However, once the contribution of the stronger target is estimated and removed from the cross-ambiguity, the weaker target could be identified as a separate peak  (Fig. \ref{fig:heatmap_iti}(b)). Likewise, for the sinc filtered waveform, we consider the case of {\em well separated targets}, where the performance is poor without ITI mitigation. Figures \ref{fig:heatmap_iti}(c) and (d) show the cross-ambiguity function heatmaps for the sinc filter before and after ITI mitigation, respectively. Here, due to the influence of the high side lobes associated with the stronger target, the peak corresponding to the weaker target is not identified without ITI mitigation (Fig.\ref{fig:heatmap_iti}(c)). However, with ITI cancellation, the weaker target could be identified as a separate peak (Fig. \ref{fig:heatmap_iti}(d)). Similar improvement is observed in the case of well separated targets for GS filter as well. While the above observations from the two-target experiment could reveal some of the key effects at play in determining sensing performance, a more detailed set of simulation results obtained by averaging over random realizations of target locations in a multi-target scenario is presented in the following section.

\begin{figure}
\label{fig:heatmap_iti}
\centering
\includegraphics[width = 8.5cm, height = 6.5cm]{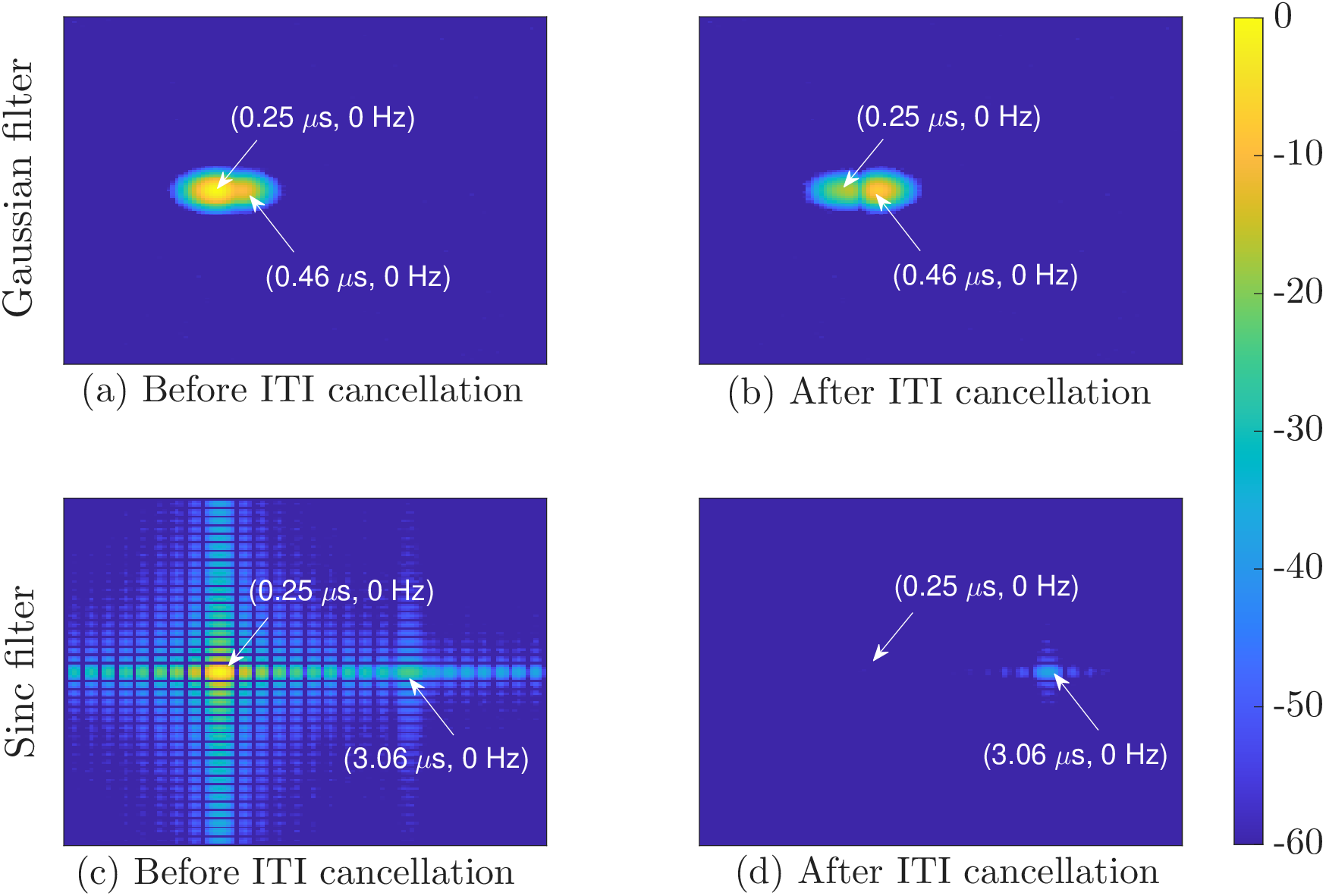}
\caption{Heatmaps illustrating the effectiveness of inter-target interference mitigation. }
\label{fig:heatmap_iti}
\end{figure}

\section{Results and Discussions}
\label{sec:results}
In this section, we present the detection and estimation performance of Zak-OTFS waveforms filtered through different DD pulse shaping filters (sinc, Gaussian, and GS filters). Two different types of radar scenes, namely, scenes with {\em densely populated targets} and scenes with {\em sparsely populated targets}, are considered. In the simulations, we consider bandwidth $B = 4$ MHz, time duration $T = 20$ ms, delay period $\tau_{\text{p}} = 100$ $\mu$s, Doppler period $\nu_{\text{p}} = 1/\tau_{\text{p}} = 10$ kHz, oversampling factors $P=Q=4$, and carrier frequency $f_{c} = 1$ GHz. While the Zak-OTFS waveform is applicable for sensing irrespective of bandwidth and time duration, the choice of $B = 4$ MHz and $T = 20$ ms can be applicable for surveillance radars \cite{zak_otfs_radar}. We consider scenes with four targets, i.e., $h_{\text{s}}(\tau,\nu) = \sum_{i =1}^{4}h_{i}\delta(\tau-\tau_{i})\delta(\nu-\nu_i)$. The fade $h_{i}$ is a random variable with complex normal distribution with zero mean and variance $\mathbb{E}[|h_{i}|^2] = (10^{-7}/(\tau_{i}))^4$. The minimum delay of the targets is  $\tau_{\text{min}} = 200$ $\mu$s corresponding to a distance of $30$ km. The four targets are uniformly distributed in DD widows $\Omega_{\text{d}}$ and $\Omega_{\text{s}}$ for densely and sparsely populated targets scenes, respectively. The DD window $\Omega_{\text{d}}$ for the scene with densely populated targets is taken as $\Omega_{\text{d}} =[200,201]$ $\mu$s  $\times$  $[-200, 200]$ Hz. The DD window $\Omega_{\text{s}}$ for the scene with sparsely populated targets is taken as $\Omega_{\text{s}} =[200,205]$ $\mu$s  $\times$  $[-1000, 1000]$ Hz. The signal to noise ratio (SNR) considered in the simulation for both scenes is the SNR for a target at the minimum delay $\tau_{\text{min}}$, i.e., $\text{SNR} = \left(\frac{10^{-7}}{\tau_{\text{min}}}\right)^4\left(\frac{E_{\text{p}}}{ BTN_{0}}\right)$, where $E_{\text{p}}$ is the energy of the probing waveform, taken to be $1$ and $N_{0}$ is the variance of AWGN. For a target at delay $\tau > \tau_{\text{min}}$, the average received SNR is $(10^{-7}/(\tau-\tau_{\text{min}}))^4$ times the SNR for a target at delay $\tau_{\text{min}}$ \cite{zak_otfs_radar}.

{\em ROC and range/velocity estimation accuracy:} To illustrate detection performance, we plot the probability of detection ($P_{D}$) as a function of the probability of false alarm ($P_{F}$), i.e., ROC curves, for different filters at a particular SNR. We define the target present hypothesis ($H_{1}$) for the $i$th target as $|A_{y,x}(\tau_{i},\nu_{i})|^2$ and the target absent hypothesis ($H_{0}$) as $|\sum_{j\neq i, j=1}^{j = 4}h_{j}A_{x,x}(\tau_{i}-\tau_{j},\nu_{i}-\nu_{j})e^{j2\pi \nu_{j}(\tau_{i}-\tau_{j})} + A_{n,x}(\tau_{i},\nu_{i})|^2$. We generate $5\times 10^{5}$ Monte Carlo instances of $H_{1}$ and $H_{0}$ for each of the targets. For a particular threshold $\gamma$, $P_{D}$ is the fraction of times $H_1 > \gamma$ averaged over all targets, and $P_{F}$ is the fraction of times $H_{0}>\gamma$ averaged over all the targets. To illustrate estimation performance, we plot the root mean square (RMS) range and velocity estimation errors of the targets as a function of SNR. 

\subsection{Scenes with densely populated targets}
Here, we consider densely populated scenes with four targets uniformly distributed in the DD window $\Omega_{\text{d}}$. Corresponding to this window, we consider the peak detection DD window as $[200,203]$ $\mu$s  $\times$  $[-600, 600]$ Hz. In Fig. \ref{fig:dense_ROC}, we plot $P_{D}$ versus $P_{F}$ at an SNR of $-9$ dB for sinc, Gaussian, and GS filters. We observe that for a particular $P_{F}$, the Gaussian filtered waveform gives a worse $P_{D}$ compared to those of the sinc and GS filtered waveforms. This behavior is attributed to the wide main lobe in its self-ambiguity function, resulting in a broader peak associated with each of the targets in the cross-ambiguity domain. As interference from other targets is significant due to the wide main lobe, the number of false detections is high, resulting in worse detection performance.  On the other hand, the self-ambiguity functions of the sinc and GS filtered waveforms have narrow main lobes, resulting in less interference due to other targets, enabling better performance compared to that of the Gaussian filtered waveform. Also, since both sinc and GS filtered waveforms have similar PSLRs, the number of false detections for both are nearly equal, resulting in similar detection performance.

\begin{figure}
\centering
\includegraphics[width = 7.5cm, height = 6cm]{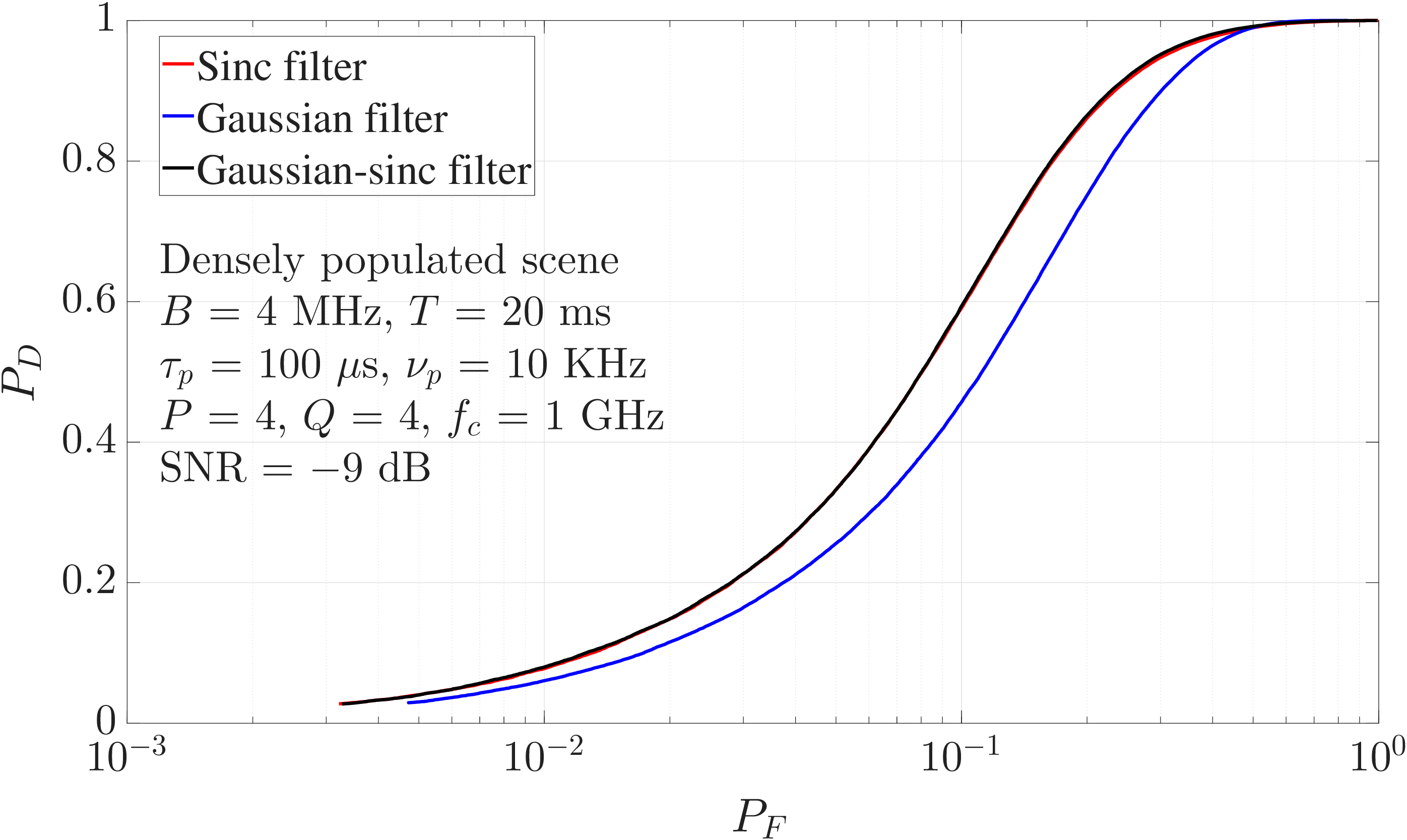}
\caption{ROC curves for different filters in densely populated scene with four targets.}
\label{fig:dense_ROC}
\end{figure}

In Fig. \ref{fig:dense_range}, we plot the RMS range estimation error for the targets as a function of SNR for different filters, for receivers without and with ITI mitigation. We observe that, with the basic peak-detection based receiver without ITI mitigation, the Gaussian filtered waveform performs much worse than the sinc and GS filtered waveforms. This is due to the broader peaks associated with each of the closely spaced targets, rendering the targets indistinguishable as different peaks. In contrast, the narrow peaks associated with the targets in the cross-ambiguity function for sinc and GS filtered waveforms allow better resolvability and performance. As both sinc and GS filtered waveforms have narrow main lobes and similar PSLRs, they give similar performance. We further observe that with ITI mitigation, the performance of all the filters improves, and this improvement is more significant for the Gaussian filter. However, the sinc and GS filters are found to perform better than the Gaussian filter after ITI mitigation as well. Similar performance trends can be observed in the velocity estimation results presented in Fig. \ref{fig:dense_vel}.

\begin{figure}
\centering
\includegraphics[width = 7.5cm, height = 6cm]{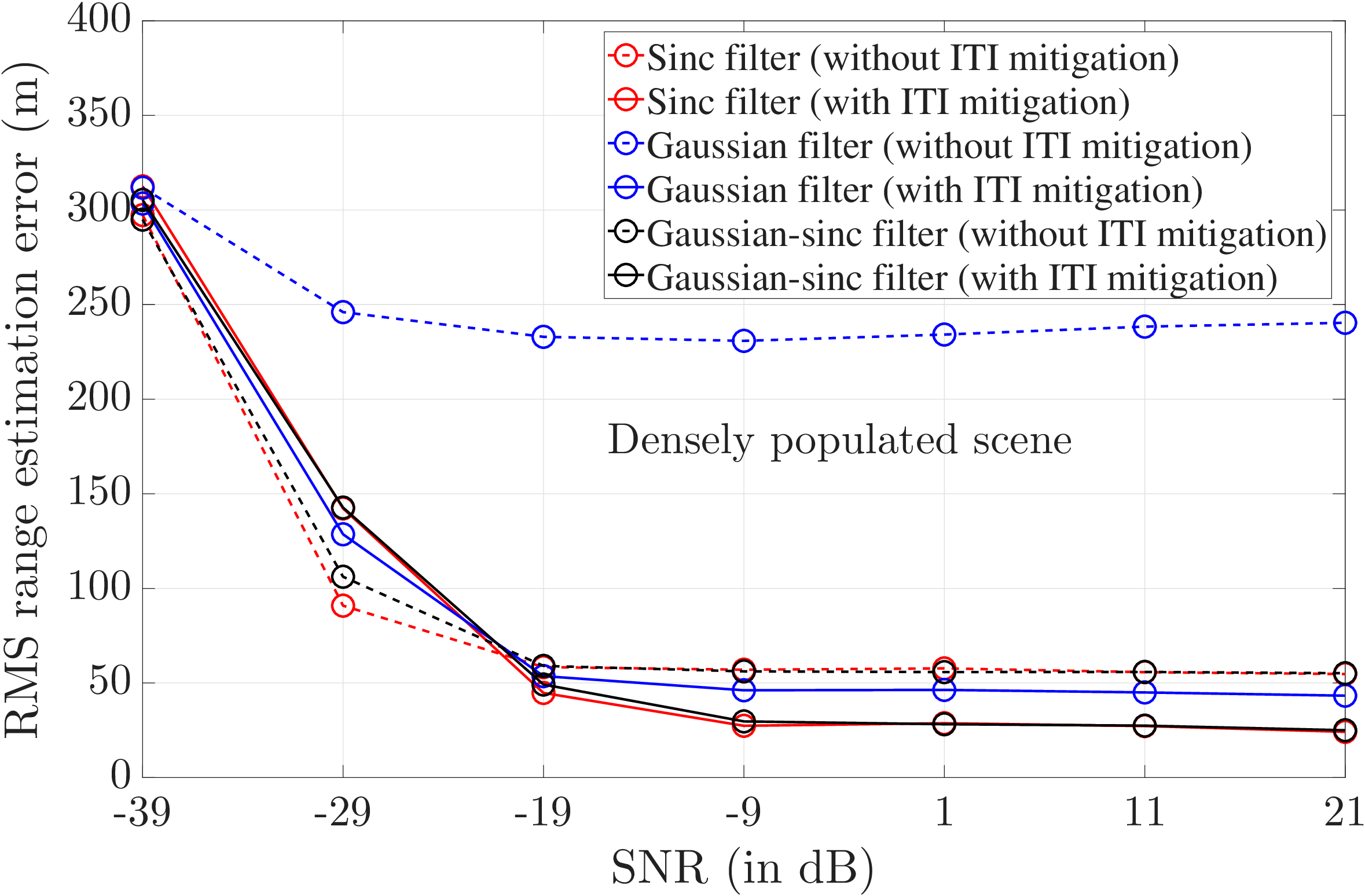}
\caption{RMS range estimation error vs SNR in densely populated scene with four targets for different filters.}
\label{fig:dense_range}
\end{figure}

\begin{figure}
\centering
\includegraphics[width = 7.5cm, height = 6cm]{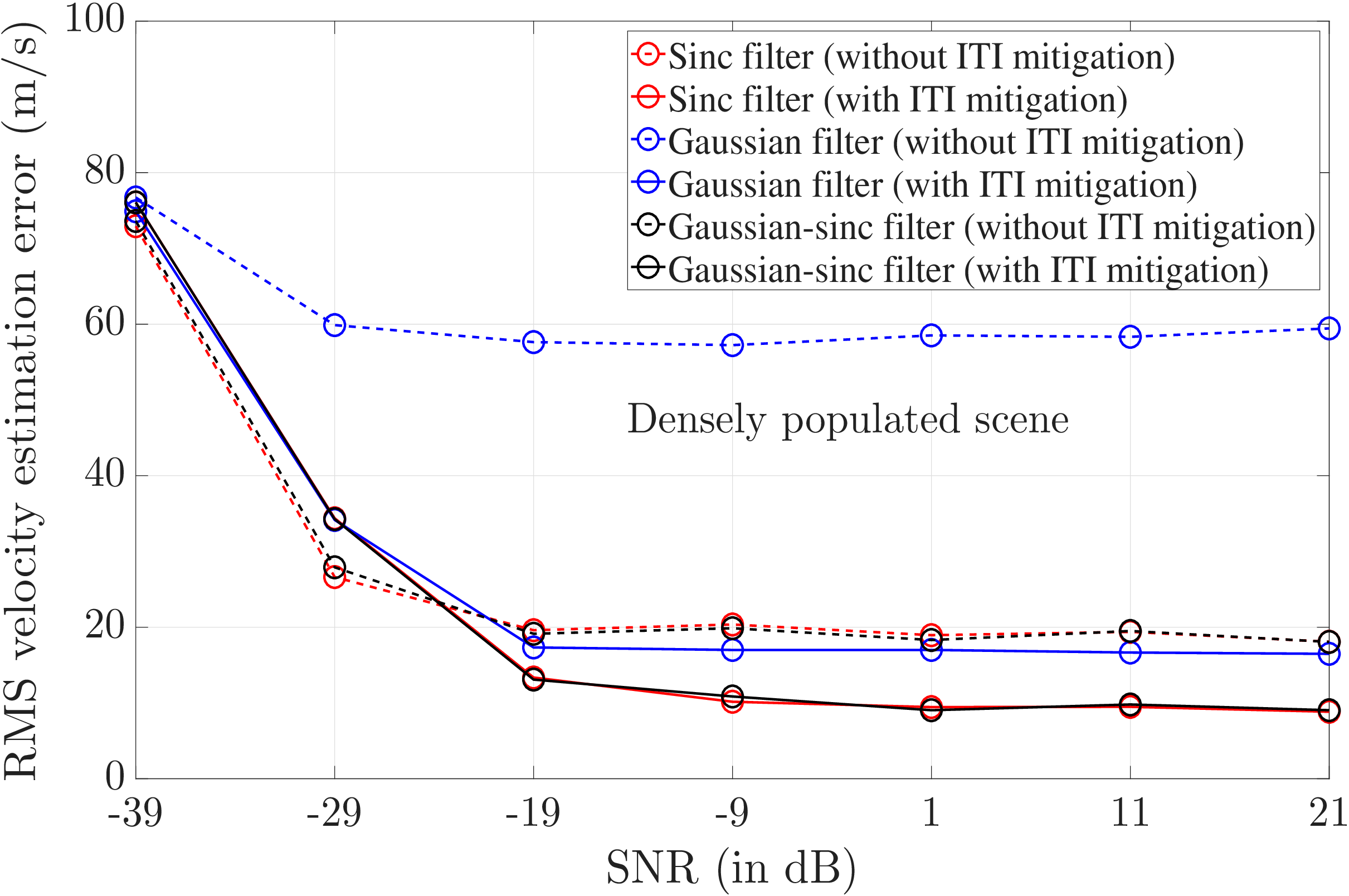}
\caption{RMS velocity estimation error vs SNR in densely populated scene with four targets for different filters.}
\label{fig:dense_vel}
\end{figure}

\subsection{Scenes with sparsely populated targets}
Here, we consider sparsely populated scenes with four targets uniformly distributed in the DD window $\Omega_{\text{s}}$. Corresponding to this window, we consider the peak detection DD window as $[200 , 207]$ $\mu$s $\times$ $[-1400,1400]$ Hz. Figure \ref{fig:sparse_ROC} shows the ROC curve ($P_{D}$ vs $P_{F}$) at an SNR of $-9$ dB. We observe that, unlike in the densely populated scenes where the Gaussian filter performed the worst, here in the sparsely populated scenes the Gaussian filter performs the best compared to the sinc and GS filters. Also, among the sinc and GS filters, sinc filter performs better. This is attributed to the fact that there are almost no sidelobes in the self-ambiguity of the Gaussian filtered waveform, which results in very low interference from other well separated targets. On the other hand, the self-ambiguity of the GS filtered waveform has an ISLR that is higher and lower than those of the Gaussian and sinc filtered waveforms, respectively. Thus, the interference encountered by a target from other well separated targets is much less in the case of the GS filter than that of the sinc filter, resulting in GS filter performing better than sinc filter.

\begin{figure}
\centering
\includegraphics[width = 7.5cm, height = 6cm]{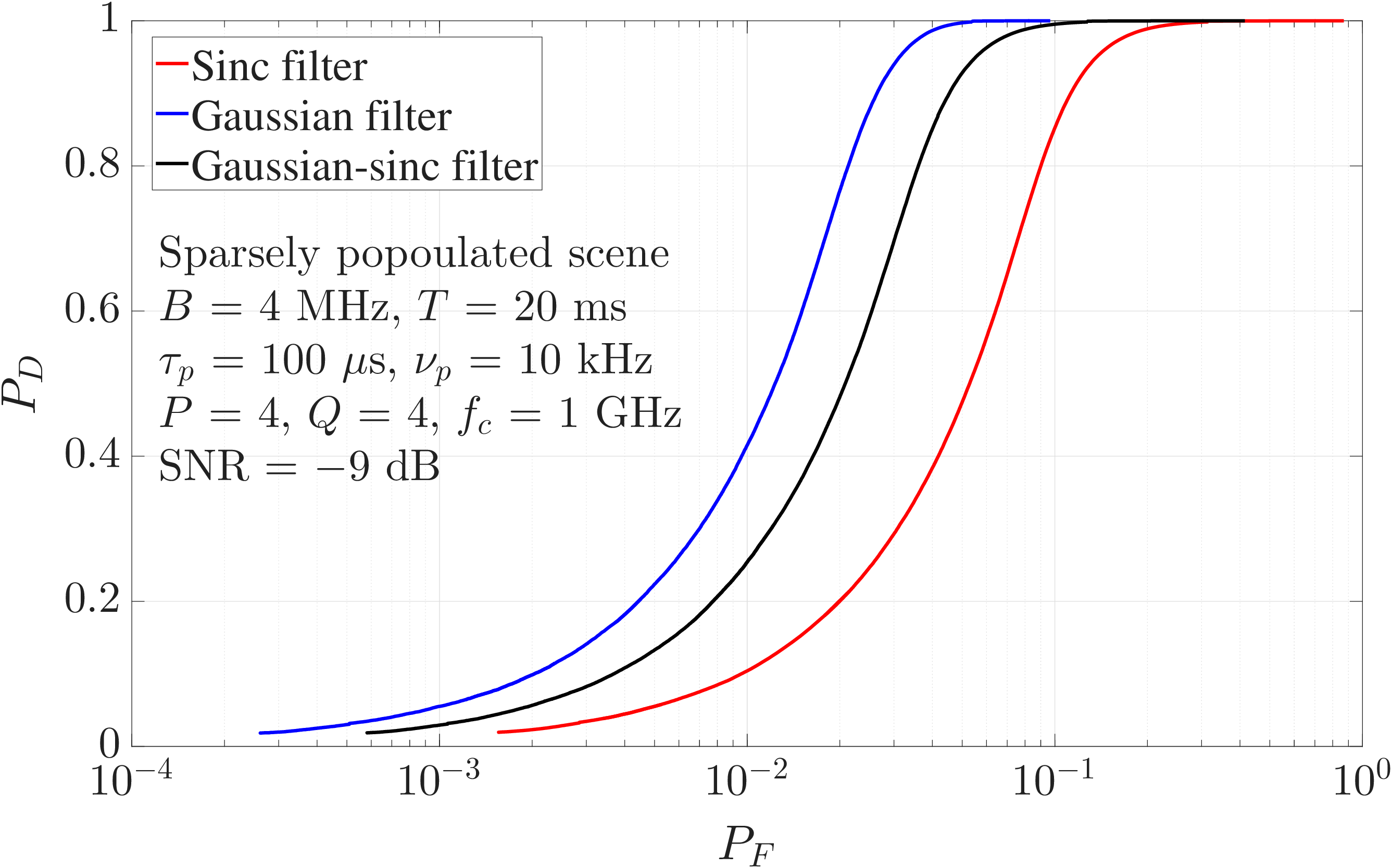}
\caption{ROC curves for different filters in sparsely populated scene with four targets.}
\label{fig:sparse_ROC}
\end{figure}

In Fig. \ref{fig:sparse_range}, we plot the RMS range estimation error performance as a function of SNR for the sparsely populated scenes. Here, we observe that, without ITI mitigation, the Gaussian filter achieves the best performance (due to very low side lobes). However, with ITI mitigation, the sinc and GS filters perform better. This is because the narrow main lobes of the sinc and GS filtered waveforms allow the interference to be estimated and canceled more effectively compared to the Gaussian filtered waveform. A similar trend can be observed in the velocity estimation performance shown in Fig. \ref{fig:sparse_vel}.

\begin{figure}
\centering
\includegraphics[width = 7.5cm, height = 6cm]{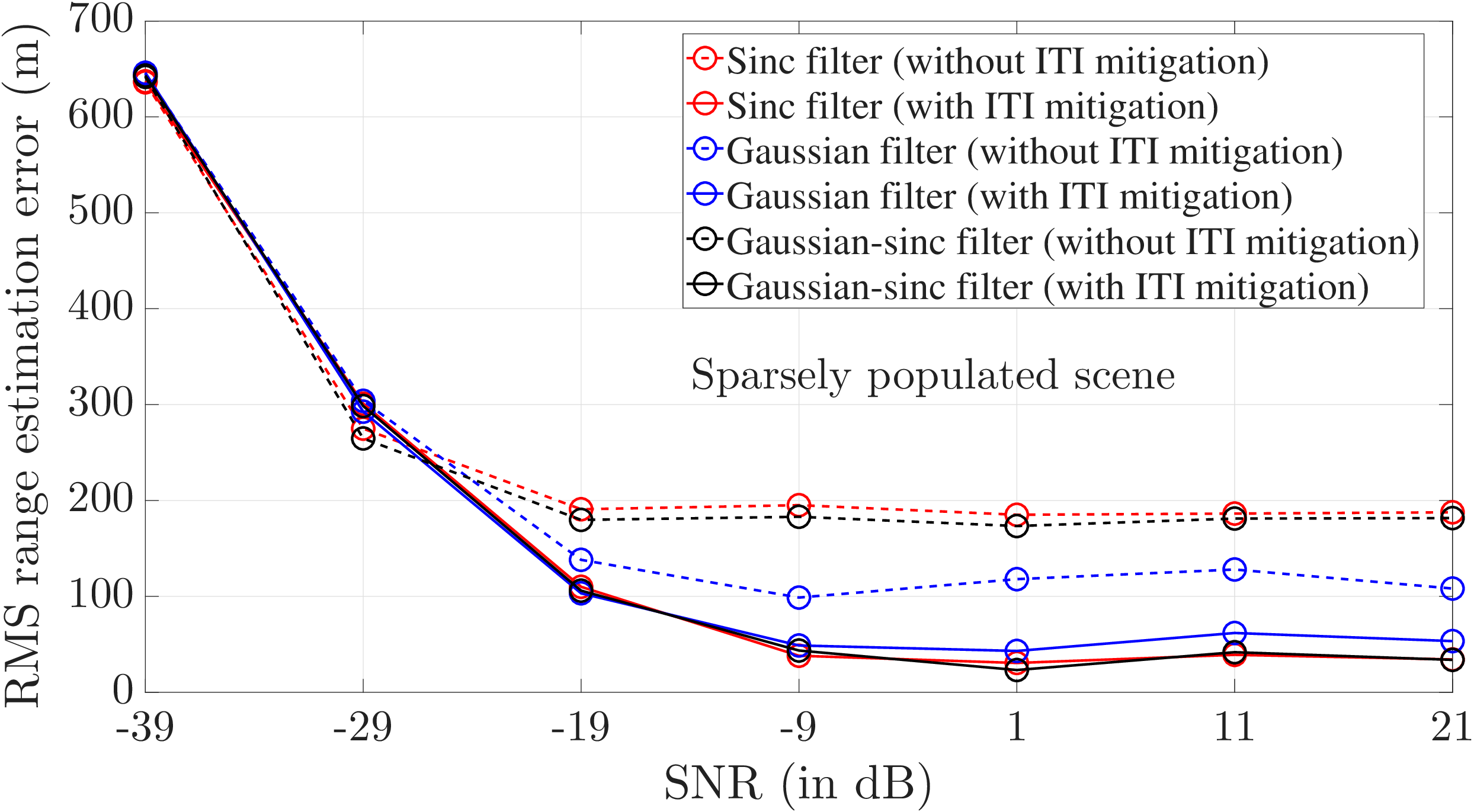}
\caption{RMS range estimation error vs SNR in sparsely populated scene with four targets for different filters.}
\label{fig:sparse_range}
\end{figure}

\begin{figure}
\centering
\includegraphics[width = 7.5cm, height = 6cm]{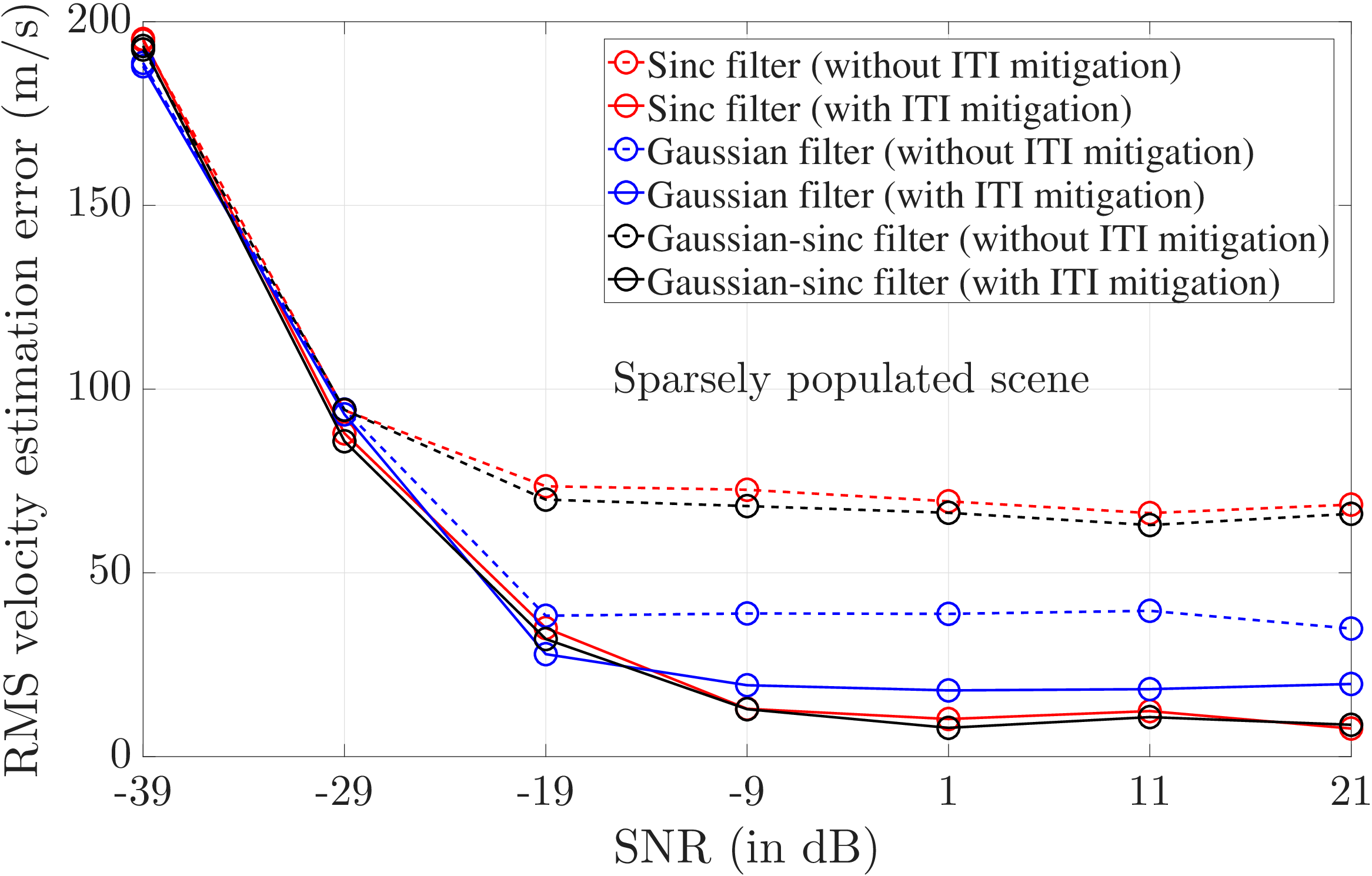}
\caption{RMS velocity estimation error vs SNR in sparsely populated scene with four targets for different filters.}
\label{fig:sparse_vel}
\end{figure}

\subsection{Summary}
\label{subsec:summary}
Based on the results/observations, we infer the following. 
\begin{itemize}
\item {\em Detection performance (ROC):}
In dense scenes, due to their narrow main lobes, sinc and GS filters give better detection performance compared to Gaussian filter. In sparse scenes, due to its very low sidelobes, Gaussian filter gives better detection performance compared to sinc and GS filters. 
\item {\em Range/velocity estimation performance:} 
When the basic peak-detection based receiver without ITI mtigation is used, the sinc and GS filters perform better in densely populated scenes due to their narrow main lobes, whereas the Gaussian filter performs better in sparsely populated scenes due to its very low sidelobes. When the ITI mitigating receiver is used, because of  effective interference estimation and mitigation enabled by their narrow main lobes, the sinc and GS filters perform better in both dense and sparsely populated scenes compared to the Gaussian filter.
\end{itemize}

\section{Conclusions}
\label{sec:concl}
We investigated the effect of DD domain pulse shaping filter on the radar sensing performance of Zak-OTFS waveform for different radar scenes (scenes with dense and sparsely populated targets) and receivers (receivers without and with ITI mitigation). Recognizing the key role that the self-ambiguity function of the waveform generated using different filters plays on the detection (ROC) and range/velocity estimation performance, we characterized the main lobe width, PSLR, and ISLR of the self-ambiguities of sinc, Gaussian, and GS filtered Zak-OTFS waveforms. Taking advantage of the analytical tractability of the Zak-OTFS framework, we derived closed-form expressions for the self-ambiguities of the sinc and GS filtered waveforms which have not been reported before. These expressions resulted in significant reduction of simulation run times compared to numerical integrations. Simulation results showed that different filters influenced the detection and range/velocity estimation performance differently, depending on the type of radar scene and receiver used, as summarized in Sec. \ref{subsec:summary} in the previous section. Investigation of filter designs that achieve  good balance of sensing and communication performance in  multi-target/multiuser environments can be taken up for future research.    

\appendices

\section{Derivation of equation (\ref{eq:noise_cov_amb})}
\label{app:derive_noise_cov}
In the derivation given in (\ref{eq:proof_cov}) at the top of next page, step (a) follows from substituting the expression for $A_{n,x}(\tau,\nu)$ given in (\ref{eq:Anx}). We use the fact that $n(t)$ is AWGN, i.e, $\mathbb{E}[n(t_{1})n^{*}(t_{2})] = N_{0}\delta(t_{1}-t_{2})$ in step (a) to arrive at step (b). Step (c) follows from the sifting property of Dirac-delta function $\delta(.)$. Subsequently, we substitute $t = t_{2}-\tau_{2}$ in the integral to get the expression in step (d). The last expression in step (e) is obtained using the expression for $A_{x,x}(\tau,\nu)$ given in (\ref{eq:self_amb_t}). This completes the derivation. 

\begin{figure*}
  \begin{eqnarray}
  \label{eq:proof_cov}
  \mathbb{E}[A_{n,x}(\tau_{1},\nu_{1})A_{n,x}^*(\tau_{1},\nu_{1})] & \overset{(a)}{=}& \int_{t_1}\int_{t_2}\mathbb{E}[n(t_1)n^*(t_{2})] x^{*}(t_{1}-\tau_{1})x(t_{2}-\tau_{2})e^{-j2\pi\nu_{1}(t_{1}-\tau_{1})} e^{j2\pi\nu_{2}(t_{2}-\tau_{2})}dt_{1}dt_{2}
  \nonumber\\
  &\overset{(b)}{=}& N_{0}\int_{t_1}\int_{t_2} \delta(t_{1}-t_{2})x^{*}(t_{1}-\tau_{1})x(t_{2}-\tau_{2})e^{-j2\pi\nu_{1}(t_{1}-\tau_{1})} e^{j2\pi\nu_{2}(t_{2}-\tau_{2})}dt_{1}dt_{2}\nonumber \\
  &\overset{(c)}{=} &N_{0}\int_{t_{2}}x^{*}(t_{2}-\tau_{1})x(t_{2}-\tau_{2})e^{-j2\pi\nu_{1}(t_{2}-\tau_{1})} e^{j2\pi\nu_{2}(t_{2}-\tau_{2})}dt_{2}\nonumber \\
  &\overset{(d)}{=}& N_{0} e^{-j2\pi\nu_{1}(\tau_{2}-\tau_{1})} e^{j2\pi(\nu_{2}-\nu_{1})(\tau_{1}-\tau_{2})}\int_{t} x(t)x^*(t-\tau_{1}+\tau_{2})e^{-j2\pi (t-\tau_{1}+\tau_{2})(\nu_{1}-\nu_{2})}dt\nonumber \\
  &\overset{(e)}{=}&N_{0} e^{-j2\pi\nu_{2}(\tau_{2}-\tau_{1})}A_{x,x}(\tau_{1}-\tau_{2},\nu_{1}-\nu_{2}).
\end{eqnarray}
\hrule
\end{figure*}

\section{Proof of Theorem \ref{thm:self_amb_gauss}}
\label{app:proof_self_amb_gauss}
Using the fact that $w_{\text{tx}}(\tau,\nu) = w_{\text{tx}}^{*}(-\tau,-\nu)$ for the Gaussian filter given in (\ref{eq:wtx_gauss}), and substituting $w_{\text{tx}}(\tau,\nu) = w_{1}(\tau)w_{2}(\nu)$ in (\ref{eq:self_amb_DD}), we get (\ref{eq:self_amb_sep}). 
\begin{figure*}[t]
\begin{eqnarray}
    \label{eq:self_amb_sep}
    A_{x,x}(\tau,\nu)=\sum_{n\in \mathbb{Z}} \ \sum_{m \in \mathbb{Z}} \ \underbrace{ \int_{\tau'} w_{1}(\tau')w_{1}(\tau-\tau'-n\tau_{\text{p}})e^{j2\pi\nu(\tau-\tau'-n\tau_{\text{p}})}}_{I_{1}(\tau,\nu)}d\tau' \underbrace{\int_{\nu'}w_{2}(\nu')w_{2}(\nu-\nu'-m\nu_{\text{p}})e^{j2\pi\nu'(n\tau_{\text{p}})}d\nu'}_{I_{2}(\tau,\nu)}.
\end{eqnarray}
\hrule
\end{figure*}
Substituting $w_{1}(\tau) = \left(\frac{2\alpha_{\tau}B^2}{\pi}\right)^{1/4}e^{-\alpha_{\tau}B^2\tau^2}$ in the $I_{1}(\tau,\nu)$ expression given in (\ref{eq:self_amb_sep}), we obtain
\begin{eqnarray}
\label{eq:gauss_I1}
    I_{1}(\tau,\nu) = \sqrt{\frac{2\alpha_{\tau}B^2}{\pi}}\int_{\tau'}e^{-\alpha_{\tau}B^2\tau'^2}e^{-\alpha_{\tau}B^2(\tau-\tau'-n\tau_{\text{p}})^2}\nonumber\\e^{j2\pi\nu(\tau-\tau'-n\tau_{\text{p}})}d\tau'.
\end{eqnarray}
Substituting $t_{0} = \tau-n\tau_{\text{p} }$ in (\ref{eq:self_amb_sep}), we get 
\begin{eqnarray}
\label{eq:gauss_I1_2}
    \hspace{-50mm}I_{1}(\tau,\nu) &=& \sqrt{\frac{2\alpha_{\tau}B^2}{\pi}} e^{-\alpha_{\tau}B^2t_{0}^2}e^{j2\pi\nu t_{0}}\nonumber\\ &&\int_{\tau'}e^{-2\alpha_{\tau}B^2\left(\tau'^2 +\frac{j2\pi\nu-2\alpha_{\tau}B^2t_{0}}{2\alpha_{\tau}B^2}\tau'\right)}d\tau'.
\end{eqnarray}
Completing the square inside the integral in (\ref{eq:gauss_I1_2}), we obtain
\begin{eqnarray}
\label{eq:gauss_I1_3}
I_{1}(\tau,\nu) \hspace{-3mm}&=\hspace{-3mm}& \sqrt{\frac{2\alpha_{\tau}B^2}{\pi}} e^{-\alpha_{\tau}B^2t_{0}^2}e^{j2\pi\nu t_{0}}e^{2\alpha_{\tau}B^2\left(\frac{j2\pi\nu-2\alpha_{\tau}B^2 t_{0}}{4\alpha_{\tau}B^2}\right)^2}\nonumber \\ &&\hspace{10mm} \int_{\tau'}e^{-2\alpha_{\tau}B^2\left(\tau'-\frac{j2\pi\nu-2\alpha_{\tau}B^2 t_{0}}{4\alpha_{\tau}B^2}\right)^2}d\tau'. 
\end{eqnarray}
Using the fact that $\int e^{-a(x-b)^2}dx = \sqrt{\frac{\pi}{a}}$, we obtain
\begin{eqnarray}
    \label{eq:gauss_I1_4}
    I_{1}(\tau,\nu) = e^{-\alpha_{\tau}B^2t_{0}^2}e^{j2\pi\nu t_{0}}e^{2\alpha_{\tau}B^2\left(\frac{j2\pi\nu-2\alpha_{\tau}B^2 t_{0}}{4\alpha_{\tau}B^2}\right)^2}.
\end{eqnarray}
Similarly, using $w_{2}(\nu) = \left(\frac{2\alpha_{\nu}T^2}{\pi}\right)^{1/4}e^{-\alpha_{\nu}T^2\nu^2}$ in the $I_{2}(\tau,\nu)$ expression given in (\ref{eq:self_amb_sep}), we obtain
\begin{eqnarray}
\label{eq:gauss_I2_1}
    \hspace{-50mm}I_{2}(\tau,\nu) &\hspace{-3mm}=& \hspace{-3mm}\sqrt{\frac{2\alpha_{\nu}T^2}{\pi}}\int_{\nu'}e^{-\alpha_{\nu}T^2\nu'^2}e^{j2\pi\nu'n\tau_{\text{p}}}\nonumber \\&& \hspace{25mm}e^{-\alpha_{\nu}T^2(\nu-\nu'-m\nu_{\text{p}})^2}d\nu'.
\end{eqnarray}
Substituting $f_{0} = \nu-m\nu_{\text{p}}$, completing the square inside the integral in (\ref{eq:gauss_I2_1}), and using the fact $\int e^{-a(x-b)^2}dx = \sqrt{\frac{\pi}{a}}$, we obtain 
\begin{eqnarray}
\label{eq:gauss_I2_2}
    I_{2}(\tau,\nu) = e^{-\alpha_{\nu}T^2f_{0}^2}e^{\frac{\alpha_{\nu}T^2}{2}\left(f_{0}+\frac{j2\pi n\tau_{\text{p}}}{2\alpha_{\nu}T^2}\right)^2}.
\end{eqnarray}
Substituting (\ref{eq:gauss_I1_4}) and (\ref{eq:gauss_I2_2}) in (\ref{eq:self_amb_sep}), we get the result in the Theorem \ref{thm:self_amb_gauss}. This completes the proof.

\section{Proof of Theorem \ref{thm:self_amb_GS}}
\label{app:proof_self_amb_GS}
Observing that (\ref{eq:self_amb_sep}) is also applicable for the GS filter given in (\ref{eq:wtx_GS}) and substituting $w_{1}(\tau) = \Omega_{\tau}\sqrt{B}\mathrm{sinc}(B\tau)e^{-\alpha_\tau B^2\tau^2}$ in the $I_{1}(\tau,\nu)$ expression given in (\ref{eq:self_amb_sep}), we obtain 
\begin{eqnarray}
\label{eq:GS_I11}
    I_{1}(\tau,\nu) = \Omega_{\tau}^2B\int_{\tau'}\mathrm{sinc}(B\tau')\mathrm{sinc}(B(\tau-\tau'-n\tau_{\text{p}}))\nonumber \\ e^{-\alpha_{\tau}B^2\tau'^2}e^{-\alpha_{\tau}B^2(\tau-\tau'-n\tau_{\text{p}})^2}e^{j2\pi\nu(\tau-\tau'-n\tau_{\text{p}})}d\tau'.
\end{eqnarray}
Let $x(\tau') = \mathrm{sinc}(B\tau')\mathrm{sinc}(B(\tau-\tau'-n\tau_{\text{p}}))$ and $y(\tau') = e^{-\alpha_{\tau}B^2\tau'^2} e^{-\alpha_{\tau}B^2(\tau-\tau'-n\tau_{\text{p}})^2}e^{j2\pi\nu(\tau-\tau'-n\tau_{\text{p}})}$. We substitute $t_{0} = \tau - n\tau_{\text{p}}$.
When $t_{0} \neq 0$,
\begin{eqnarray}
X^{*}(f) &=& \left(\int x(\tau')e^{-j2\pi f\tau'}d\tau'\right)^{*} \nonumber \\&\hspace{-20mm}=&\hspace{-12mm} \frac{1}{j2\pi B^2 t_{0}}\left(-e^{-j2\pi(\frac{B}{2}-f)t_{0}}+e^{j\pi B t_{0}}\right) \mathbb{I}_{0<f<B} \ \nonumber \\&\hspace{-15mm}&\hspace{-13mm}+ \frac{1}{j2\pi B^2 t_{0}}\left(e^{j2\pi(\frac{B}{2}+f)t_{0}}-e^{-j\pi B t_{0}}\right) \mathbb{I}_{-B<f<0},\nonumber
\end{eqnarray}
 and
\begin{eqnarray}
Y(f) &=& \int y(\tau')e^{-j2\pi f\tau'}d\tau' \nonumber \\&\hspace{-17mm}=&\hspace{-10mm}\sqrt{\frac{\pi}{2\alpha_{\tau}B^2}}e^{-\frac{\alpha_{\tau}B^2}{2}t_{0}^2}e^{j2\pi \nu t_{0}}e^{-\frac{\pi^2(f+\nu)^2}{2\alpha_{\tau}B^2}}e^{-j\pi(f+\nu)t_{0}}.   \nonumber
\end{eqnarray}
Applying Parseval's theorem, i.e., $\int x^*(\tau')y(\tau')d\tau' = \int X^{*}(f)Y(f)df$ in (\ref{eq:GS_I11}),  we get $I_{1,1}(\tau,\nu) = I_{1,1,1}(\tau,\nu)+I_{1,1,2}(\tau,\nu)$, where $I_{1,1,1}(\tau,\nu)$ and $I_{1,1,2}(\tau,\nu)$ are given by (\ref{eq:GS_I111}) and (\ref{eq:GS_I112}), respectively. 
In (\ref{eq:GS_I111}), to arrive at the expression for $I_{1,1,1}(\tau,\nu)$, we substitute $P_{1} =  \frac{\Omega_{\tau}^2}{j2\pi t_{0} B} \sqrt{\frac{\pi}{2\alpha_{\tau}B^2}}e^{-\frac{\alpha_\tau B^2 t_{0}^2}{2}}$ to get the expression in step (a). We complete the square inside the integrals in the step (a) and substitute $K_{1} = \frac{\alpha_{\tau}B^2}{\pi^2}\left(\frac{\pi^2 \nu}{\alpha_{\tau} B^2} + j\pi t_{0}\right)$ and $K_{2} = \frac{\alpha_{\tau}B^2}{\pi^2}\left(\frac{\pi^2 \nu}{\alpha_{\tau} B^2} - j\pi t_{0}\right)$, to arrive at step (b). We use the the function $g_{1}(.)$ to obtain the expression in step (c), where $g_{1}(a,t,s) = \int_{s}^{t}e^{-ax^2}dx = \frac{\sqrt{\pi}}{2\sqrt{a}}\left({erf}(\sqrt{a}t) - {erf}(\sqrt{a}s)\right)$ for complex numbers $a,t,s$.
\begin{figure*}[t]
\begin{eqnarray}
\label{eq:GS_I111}
I_{1,1,1}(\tau,\nu) &=&\frac{\Omega_{\tau}^2}{j 2\pi B t_{0}} \sqrt{\frac{\pi}{2\alpha_{\tau}B^2}}e^{-\frac{\alpha_{\tau}B^2 t_{0}^2}{2}} \int_{0}^{B}\left(e^{j\pi B t_{0}}-e^{-j2\pi\left(\frac{B}{2}-f\right)t_{0}}\right)e^{j2\pi \nu t_{0}}e^{-\frac{\pi^2 (f+\nu)^2}{2 \alpha_{\tau}B^2}}e^{-j \pi (f+\nu)t_{0}}df\nonumber\\
&\hspace{-38mm}\overset{(a)}{=}&\hspace{-20mm} P_{1}e^{-\frac{\pi^2 \nu^2}{2\alpha_{\tau}B^2}}e^{j\pi\nu t_{0}}\left[e^{j\pi Bt_{0}} \int_{0}^{B} e^{-\frac{\pi^2}{2\alpha_{\tau}B^2}\left(f^2 +\frac{2\alpha_{\tau}B^2}{\pi^2}\left(\frac{2\pi^2\nu}{2\alpha_{\tau}B^2}+j\pi t_{0}\right)f\right)} - e^{-j\pi Bt_{0}} \int_{0}^{B}e^{-\frac{\pi^2}{2\alpha_{\tau}B^2}\left(f^2 +\frac{2\alpha_{\tau}B^2}{\pi^2}\left(\frac{2\pi^2\nu}{2\alpha_{\tau}B^2}-j\pi t_{0}\right)f\right)}\right]\nonumber\\
&\hspace{-38mm}\overset{(b)}{=}&\hspace{-20mm} P_{1}e^{-\frac{\pi^2 \nu^2}{2\alpha_{\tau}B^2}}e^{j\pi\nu t_{0}}\left[e^{j\pi Bt_{0}}e^{\frac{\pi^2} {2\alpha_{\tau}B^2}K_{1}^2}\int_{0}^{B}e^{-\frac{\pi^2}{2\alpha_{\tau}B^2}\left(f+K_{1}\right)^2}df - e^{-j\pi Bt_{0}}e^{\frac{\pi^2} {2\alpha_{\tau}B^2}K_{2}^2}\int_{0}^{B}e^{-\frac{\pi^2}{2\alpha_{\tau}B^2}\left(f+K_{2}\right)^2}df\right]\nonumber \\
&\hspace{-38mm}\overset{(c)}{=}&\hspace{-20mm} P_{1}e^{-\frac{\pi^2 \nu^2}{2\alpha_{\tau}B^2}}e^{j\pi\nu t_{0}}\left[e^{j\pi Bt_{0}}e^{\frac{\pi^2} {2\alpha_{\tau}B^2}K_{1}^2}g_{1}\left(\frac{\pi^2}{2\alpha_{\tau}B^2},B+K_{1},K_{1}\right) -e^{-j\pi Bt_{0}}e^{\frac{\pi^2} {2\alpha_{\tau}B^2}K_{2}^2}g_{1}\left(\frac{\pi^2}{2\alpha_{\tau}B^2},B+K_{2},K_{2}\right)  \right].
\end{eqnarray}
\hrule
\end{figure*}
\begin{figure*}[t]
\begin{eqnarray}
\label{eq:GS_I112}
\hspace{-2mm}
I_{1,1,2}(\tau,\nu) &=& \frac{\Omega_{\tau}^2}{j 2\pi B t_{0}} \sqrt{\frac{\pi}{2\alpha_{\tau}B^2}}e^{-\frac{\alpha_{\tau}B^2 t_{0}^2}{2}} \int_{-B}^{0} \left(e^{j2\pi\left(\frac{B}{2}+f\right)t_{0}}-e^{-j\pi Bt_{0}}\right)e^{j2\pi \nu t_{0}}e^{-\frac{\pi^2 (f+\nu)^2}{2 \alpha_{\tau}B^2}}e^{-j \pi (f+\nu)t_{0}}df\nonumber\\
&\hspace{-44mm}\overset{(a)}{=}&\hspace{-24.5mm} P_{1}e^{-\frac{\pi^2 \nu^2}{2\alpha_{\tau}B^2}}e^{j\pi\nu t_{0}}\left[e^{j\pi Bt_{0}}e^{\frac{\pi^2} {2\alpha_{\tau}B^2}K_{2}^2}\int_{-B}^{0}e^{-\frac{\pi^2}{2\alpha_{\tau}B^2}\left(f+K_{2}\right)^2}df - e^{-j\pi Bt_{0}}e^{\frac{\pi^2} {2\alpha_{\tau}B^2}K_{1}^2}\int_{-B}^{0}e^{-\frac{\pi^2}{2\alpha_{\tau}B^2}\left(f+K_{1}\right)^2}df\right]\nonumber \\
&\hspace{-44mm}\overset{(b)}{=}&\hspace{-24.5mm} P_{1}e^{-\frac{\pi^2 \nu^2}{2\alpha_{\tau}B^2}}e^{j\pi\nu t_{0}}\left[e^{j\pi Bt_{0}}e^{\frac{\pi^2} {2\alpha_{\tau}B^2}K_{2}^2}g_{1}\left(\frac{\pi^2}{2\alpha_{\tau}B^2},K_{2},-B+K_{2}\right) -e^{-j\pi Bt_{0}}e^{\frac{\pi^2} {2\alpha_{\tau}B^2}K_{1}^2}g_{1}\left(\frac{\pi^2}{2\alpha_{\tau}B^2},K_{1},-B+K_{1}\right)  \right]\hspace{-1mm}.
\end{eqnarray}
\hrule
\end{figure*}
Similarly, in the expression for $I_{1,1,2}(\tau,\nu) $ in (\ref{eq:GS_I112}), we substitute for $P_{1}$, $K_{1}$ and $K_{2}$ to get the expression in step (a). Thereafter, we complete the square inside the integrals to obtain the last expression in step (b).
When $t_{0}=0$, 
\begin{eqnarray}
    X^{*}(f) = \frac{B-f}{B^2}\mathbb{I}_{0<f<B} + \frac{B+f}{B^2}\mathbb{I}_{-B<f<0},
\end{eqnarray}
and
\begin{eqnarray}
        Y(f) = \sqrt{\frac{\pi}{2\alpha_{\tau}B^2}}e^{-\frac{\pi^2(f+\nu)^2}{2\alpha_{\tau}B^2}}.
\end{eqnarray}
Applying Parseval's theorem in (\ref{eq:GS_I11}) and substituting for $X^*(f)$ and $Y(f)$ in $I_{1,2}(\tau,\nu) = \int X^{*}(f)Y(f)df$, we get
\begin{eqnarray}
\label{eq:GS_I12_1}
    I_{1,2}(\tau,\nu) 
    =\Omega_{\tau}^2 B\sqrt{\frac{\pi}{2\alpha_{\tau}B^2}} 
          \left[\int_{0}^B\frac{B-f}{B^2}e^{-\frac{\pi^2(f+\nu)^2}{2\alpha_{\tau}B^2}}df \right. \nonumber \\+\left.  \int_{-B}^0\frac{B+f}{B^2}e^{-\frac{\pi^2(f+\nu)^2}{2\alpha_{\tau}B^2}}df\right].      
\end{eqnarray}
Substituting $x = f+\nu$, we obtain 
\begin{eqnarray}
\label{eq:GS_I12_2}
    I_{1,2}(\tau,\nu) 
    &=&\frac{\Omega_{\tau}^2}{B} \sqrt{\frac{\pi}{2\alpha_{\tau}B^2}} 
          \nonumber \\ &&\hspace{-15mm}\left[(B+\nu)\int_{\nu}^{B+\nu}e^{-\frac{\pi^2x^2}{2\alpha_{\tau}B^2}}dx   - \int_{\nu}^{B+\nu}xe^{-\frac{\pi^2x^2}{2\alpha_{\tau}B^2}}dx \right. \nonumber \\&& \hspace{-25mm}+\left.  \int_{-B+\nu}^{\nu}xe^{-\frac{\pi^2x^2}{2\alpha_{\tau}B^2}}dx + (B-\nu)\int_{-B+\nu}^{\nu}e^{-\frac{\pi^2x^2}{2\alpha_{\tau}B^2}}dx\right]. 
\end{eqnarray}
We define the functions $g_{1}(.)$ and $g_2(.)$ as $g_{1}(a,t,s) = \int_{s}^{t}e^{-ax^2}dx = \frac{\sqrt{\pi}}{2\sqrt{a}}\left(erf(\sqrt{a}t) - erf(\sqrt{a}s)\right)$ and $g_{2}(a,t,s) = \int_s^{t}xe^{-ax^2}dx= \frac{1}{2a}\left(e^{-as^2}-e^{-at^2}\right)$, respectively. Using these functions for (\ref{eq:GS_I12_2}), we get the expression of $I_{1,2}(\tau,\nu)$ given in Theorem \ref{thm:self_amb_GS}. Thus, we can write $I_{1}(\tau,\nu) = I_{1,1}(\tau,\nu)\mathbb{I}_{t_{0}\neq 0} + I_{1,2}(\tau,\nu)\mathbb{I}_{t_{0}=0}$.

Now, substituting $w_{2}(\nu) = \Omega_{\nu}\sqrt{T}\mathrm{sinc}(T\nu)e^{-\alpha_\nu T^2\nu^2}$ in the $I_{2}(\tau,\nu)$ expression given in (\ref{eq:self_amb_sep}), we obtain
\begin{eqnarray}
\label{eq:GS_I2_1}
    I_{2}(\tau,\nu) = \Omega_{\nu}^2T\int_{\nu'}\mathrm{sinc}(T\nu')\mathrm{sinc}(T(\nu-\nu'-m\nu_{\text{p}}))\nonumber \\ e^{-\alpha_{\nu}T^2\nu'^2}e^{-\alpha_{\nu}T^2(\nu-\nu'-m\nu_{\text{p}})^2}e^{j2\pi\nu'n\tau_{\text{p}}}d\tau'.
\end{eqnarray}
Let $Y(\nu') = \mathrm{sinc}(T\nu')\mathrm{sinc}(T(\nu-m\nu_{\text{p}}))$ and $X(\nu') = e^{-\alpha_{\nu}T^2\nu'^2}e^{-\alpha_{\nu}T^2(\nu-\nu'-m\nu_{\text{p}})^2}e^{j2\pi\nu'n\tau_{\text{p}}}$. We substitute $f_{0} = \nu-m\nu_{\text{p}}$. When $f_{0}\neq 0$,
\begin{eqnarray}
    x(t) &=& \int X(\nu')e^{j2\pi\nu't}d\nu' \nonumber \\&\hspace{-20mm}=&\hspace{-13mm} \sqrt{\frac{\pi^2}{2\alpha_{\nu}T^2}}e^{-\frac{\alpha_{\nu}T^2f_{0}^2}{2}}e^{-\frac{\pi^2}{2\alpha_{\nu}T^2}(t+n\tau_{\text{p}})^2}e^{j\pi f_{0}(t+n\tau_{\text{p}})},\nonumber
\end{eqnarray}
and
\begin{eqnarray}
    y^{*}(t) &=& \left(\int Y(\nu')e^{j2\pi \nu' t}d\nu'\right)^{*} \nonumber \\
    &\hspace{-22mm}=&\hspace{-13mm}\frac{1}{j2\pi f_{0}T^2}\left(e^{j2\pi f_{0} \left(\frac{T}{2}-t\right)}-e^{-j\pi f_{0}T}\right)\mathbb{I}_{0<t<T}\nonumber \\
    &\hspace{-15mm}+& \hspace{-10mm} \frac{1}{j2\pi f_{0}T^2}\left(e^{j\pi f_{0}T}-e^{j2\pi f_{0}(-\frac{T}{2}-t)}\right)\mathbb{I}_{-T<t<0}.\nonumber
\end{eqnarray}
Applying Parseval's theorem, i.e., $\int X(\nu')Y^{*}(\nu')d\nu' = \int x(t)y^{*}(t)dt$ in (\ref{eq:GS_I2_1}), and substituting $x = t+n\tau_{p}$, we obtain (\ref{eq:GS_I21}). We use the function $f(.)$ to arrive at the expression of $I_{2,1}$ given in Theorem \ref{thm:self_amb_GS}, where $f(t,s,z,a) = \int_{s}^{t}e^{-ax^2}e^{-jzx}dx = \frac{\sqrt{\pi}}{2}\frac{e^{-z^2/4a}}{\sqrt{a}}\left[erf\left(\sqrt{a} \ t + j\frac{z}{2\sqrt{a}}\right)-erf\left(\sqrt{a} \ s + j\frac{z}{2\sqrt{a}}\right)\right]$ for complex numbers $t,s,z,a$. 
\begin{figure*}[t]
\begin{eqnarray}
\label{eq:GS_I21}
    I_{2,1}(\tau,\nu) = \frac{\Omega_{\nu}^2}{j2\pi f_{0} T} \sqrt{\frac{\pi}{2\alpha_{\nu}T^2}}e^{-\frac{\alpha_{\nu}T^2 f_{0}^2}{2}} 
    \left[-e^{-j\pi f_{0}T}\int_{n\tau_{\text{p}}}^{T+n\tau_{\text{p}}}e^{-\frac{\pi^2}{2\alpha_{\nu}T^2}x^2}e^{j\pi f_{0}x}dx 
   +e^{j\pi f_{0}T}\int_{-T+n\tau_{\text{p}}}^{n\tau_{\text{p}}}e^{-\frac{\pi^2}{2\alpha_{\nu}T^2}x^2}e^{j\pi f_{0}x}dx\right. \nonumber \\ \left.
     -e^{j2\pi f_{0}\left(-\frac{T}{2}+n\tau_{\text{p}}\right)}\int_{-T+n\tau_{\text{p}}}^{n\tau_{\text{p}}}e^{-\frac{\pi^2}{2\alpha_{\nu}T^2}x^2}e^{-j\pi f_{0}x}dx 
     +e^{j2\pi f_{0}\left(\frac{T}{2}+n\tau_{\text{p}}\right)}\int_{n\tau_{\text{p}}}^{T+n\tau_{\text{p}}}e^{-\frac{\pi^2}{2\alpha_{\nu}T^2}x^2}e^{-j\pi f_{0}x}dx\right].  
\end{eqnarray}
\hrule
\end{figure*}
\begin{figure*}[t]
\begin{eqnarray}
\label{eq:GS_I22}
    I_{2,2}(\tau,\nu) = \frac{\Omega_{\nu}^2}{ T} \sqrt{\frac{\pi}{2\alpha_{\nu}T^2}} 
    \left[(T+n\tau_{\text{p}})\int_{n\tau_{\text{p}}}^{T+n\tau_{\text{p}}}e^{-\frac{\pi^2}{2\alpha_{\nu}T^2}x^2}dx 
   -\int_{n\tau_{\text{p}}}^{T+n\tau_{\text{p}}}xe^{-\frac{\pi^2}{2\alpha_{\nu}T^2}x^2}dx\right. \nonumber \\ \left.
     +(T-n\tau_{\text{p}})\int_{-T+n\tau_{\text{p}}}^{n\tau_{\text{p}}}e^{-\frac{\pi^2}{2\alpha_{\nu}T^2}x^2}dx 
     +\int_{-T+n\tau_{\text{p}}}^{n\tau_{\text{p}}}xe^{-\frac{\pi^2}{2\alpha_{\nu}T^2}x^2}\right].  
\end{eqnarray}
\hrule
\end{figure*}
When $f_{0} = 0$, 
\begin{eqnarray}
     x(t) = \sqrt{\frac{\pi^2}{2\alpha_{\nu}T^2}}e^{-\frac{\pi^2}{2\alpha_{\nu}T^2}(t+n\tau_{\text{p}})^2},\nonumber
\end{eqnarray}
\begin{eqnarray}
       y^{*}(t) = 
    \frac{T-t}{T^2}\mathbb{I}_{0<t<T} + \frac{T+t}{T^2}\mathbb{I}_{-T<t<0}.\nonumber
\end{eqnarray}
Applying Parseval's theorem in (\ref{eq:GS_I2_1}) and substituting $x = t+n\tau_{p}$, we obtain (\ref{eq:GS_I22}). Using the functions $g_{1}(.)$ and $g_{2}(.)$ defined earlier to prove the expression for $I_{1,2}(\tau,\nu)$, we get the expression for $I_{2,2}(\tau,\nu)$ given in Theorem \ref{thm:self_amb_GS}. Finally, we get $I_{2}(\tau,\nu) = I_{2,1}(\tau,\nu)\mathbb{I}_{f_{0}\neq 0} + I_{2,2}(\tau,\nu)\mathbb{I}_{f_{0}=0}$. Using the expressions obtained  for $I_{1}(\tau,\nu)$ and $I_{2}(\tau,\nu)$ in (\ref{eq:self_amb_sep}), we get the result in Theorem \ref{thm:self_amb_GS}. This completes the proof.

The $erf(.)$ functions in the derived expressions can be computed using accurate closed-form approximations for the $erf(.)$ function \cite{erf_1},\cite{erf_2}.

\end{document}